\documentclass[journal]{IEEEtran}
\usepackage{fancyhdr}
\usepackage{amsmath,amssymb,amsthm}
\usepackage{amsfonts}
\usepackage[dvips]{graphicx}
\usepackage{dsfont}
\usepackage{comment}

\usepackage[hidelinks]{hyperref}    

\usepackage{color}
\usepackage{pgfplots}           
\pgfplotsset{compat=1.16}
\usepackage{enumerate}
\usepackage{graphics}
\usepackage{subfigure}
\usepackage{cite}
\usepackage{mathrsfs}
\usepackage{bbm}
\usepackage{stfloats}

\usepgfplotslibrary{fillbetween}
\usetikzlibrary{patterns}
\usetikzlibrary {patterns,patterns.meta}

\usepackage{tikz}
\usepackage{mathdots}
\usepackage{yhmath}
\usepackage{cancel}
\usepackage{color}
\usepackage{siunitx}
\usepackage{array}
\usepackage{multirow}
\usepackage{textcomp}
\usepackage{gensymb}
\usepackage{tabularx}
\usepackage{booktabs}
\usepackage{graphicx}
\usepackage{xcolor}
\usepackage[utf8]{inputenc}
\usepackage[T1]{fontenc}
\usepackage{algorithm}
\usepackage{algpseudocode}

\newtheorem{lemma}{Lemma}

\renewcommand{\dbinom}[2]{\left(\!\!\begin{array}{c}{#1} \\ {#2} \end{array}\!\!\right)}
\let\oldbinom\binom
\renewcommand{\binom}[2]{\mathchoice{\dbinom{#1}{#2}}{\oldbinom{#1}{#2}}{\oldbinom{#1}{#2}}{\oldbinom{#1}{#2}}}

\allowdisplaybreaks

\begin{document}
\title{Breaking Network Densification Limits with Distributed Cooperative Massive Access (DCMA)}
\author{
    Christoforos I. Dallas, Apostolos A. Tegos, Sotiris A. Tegos,~\IEEEmembership{Senior Member,~IEEE,} Yue Xiao,
    \\
    George K. Karagiannidis,~\IEEEmembership{Fellow,~IEEE} and Panagiotis D. Diamantoulakis,~\IEEEmembership{Senior Member,~IEEE}
    \thanks{C. I. Dallas and A. A. Tegos are with the Department of Electrical and Computer Engineering, Aristotle University of Thessaloniki, 54124 Thessaloniki, Greece (e-mails: christofd@ece.auth.gr, apotegath@auth.gr).}
    \thanks{S. A. Tegos, G. K. Karagiannidis and P. D. Diamantoulakis are with the Department of Electrical and Computer Engineering, Aristotle University of Thessaloniki, 54124 Thessaloniki, Greece and with the Provincial Key Laboratory of Information Coding and Transmission, Southwest Jiaotong University, Chengdu 610031, China (e-mails: tegosoti@auth.gr, geokarag@auth.gr, padiaman@auth.gr).}
    \thanks{Yue Xiao is with the Provincial Key Laboratory of Information Coding and Transmission, Southwest Jiaotong University, Chengdu 610031, China (e-mail: alice\_xiaoyue@hotmail.com)}
    \vspace{-3mm}
  }

\maketitle
\begin{abstract}   
    In this work, we investigate the performance of the distributed cooperative massive access (DCMA) framework in large-scale network setups by incorporating stochastic geometry modeling. A partially centralized cell-free cloud-radio access (C-RAN) architecture is considered where remote radio heads (RRHs) decode transmitted messages and cooperate with each other to enhance system performance. Specifically, they can share decoded messages via feedback links, allowing receivers to cancel inter-user interference through successive interference cancellation (SIC), thus improving the decoding capabilities of the system. For such a network, we propose a novel synergetic decoding algorithm that efficiently resolves the assignment and message sharing routing for each user while accounting for practical network constraints. Furthermore, using game theory, we develop a merge-and-split algorithm with lexicographic preference to solve the problem of minimizing the RRHs utilized without compromising the performance. Simulation results show that the proposed framework significantly outperforms systems that do not implement SIC or take advantage of the cooperation between RRHs in terms of outage probability. Finally, we evaluate the performance of the proposed algorithms and validate their efficiency.    
\end{abstract}


\section{Introduction}
One of the key challenges in beyond-5G and sixth generation (6G) wireless networks is to achieve high spectral efficiency, fairness, and reliability in ultra-dense and interference-limited environments \cite{6G-Karagiannidis}. Fulfilling these increasing requirements necessitates fundamental improvements to current cellular infrastructures  and development of advanced multiple access schemes \cite{NGMA}. To this end, non-orthogonal multiple access (NOMA) stands out as a pivotal design paradigm for beyond-5G wireless networks \cite{NOMA-5G-Beyond} and can be a crucial foundation of next generation multiple access (NGMA) for 6G \cite{NOMA-to-NGMA}. Unlike conventional orthogonal multiple access (OMA), NOMA can substantially enhance both spectral efficiency and overall network connectivity \cite{5G-NOMA-Karagiannidis}. The key idea behind NOMA is to serve multiple users in the same resource block, such as time, frequency or spreading code and for this reason successive interference cancellation (SIC) is a very important technique, since it allows receivers to mitigate the interference from already decoded users. In power-domain NOMA (PD-NOMA), specifically, multiple users are allowed to share the same resource block while their separation occurs in the power domain with the aid of superposition coding (SC) and SIC. In a downlink scenario, the base station multiplexes the messages of multiple users into a single composite signal via SC. To ensure user fairness, power allocation is usually performed inversely proportional to the channel quality, where users with poor channel conditions are assigned higher transmit power, whereas users with strong channels receive lower power. At the receiver side, the strong user employs SIC to first decode the weak user's signal, reconstruct it, subtract it from the aggregate received wave, and subsequently decode its own signal. Conversely, the weak user directly decodes its own message by treating the strong user's signal as noise. In the uplink scenario, multiple users transmit their individual signals simultaneously over the shared resource block and unlike downlink, no SC is performed at the transmitters. Instead, the receiving node observes a naturally superposed aggregate signal. Due to diverse path losses and independent channel fading, these independent transmissions arrive at the receiver with distinct power levels. The receiver exploits this received power disparity to execute an iterative SIC cascade. Furthermore, there is no specific decoding order, but a common practice is to decode the users in descending order starting from the strongest. 

The flexibility of PD-NOMA in the uplink has been leveraged to resolve a diverse range of challenges in beyond-5G networks. At the fundamental physical layer, foundational research has focused on optimal power allocation strategies to maximize the uplink sum-rate capacity and spectral efficiency \cite{PD-NOMA-power-allocation}. To further mitigate intra-cell interference during the SIC cascade, dynamic user clustering algorithms have been jointly optimized with power control to ensure reliable decoding \cite{PD-NOMA-user-clustering}. Because superimposed transmissions over a shared resource block create vulnerabilities to eavesdropping, physical layer security (PLS) has been enhanced by deploying energy-harvesting jammers that emit artificial noise to protect legitimate uplink NOMA users \cite{PD-NOMA-PLS}. Furthermore, to address the latency limitations of next-generation devices, NOMA has been linked with mobile edge computing (MEC), drastically reducing the delay and energy consumption associated with computational task offloading \cite{PD-NOMA-MEC}. Finally, in the domain of sustainability, advanced resource management schemes have been developed for NOMA-based wireless-powered communication networks, optimizing time and power allocation to maximize overall energy efficiency \cite{PD-NOMA-EE}.

\subsection{Literature Review}
Although single base station (BS) uplink NOMA has been extensively investigated, distributed uplink NOMA (DU-NOMA), where distributed remote radio heads (RRHs) serve multiple users within the same resource block, has received significantly less research attention. The foundational study in \cite{CoMP} highlights the integration of NOMA within coordinated multi-point (CoMP) networks as a highly viable strategy for beyond-5G architectures. Specifically, this framework promotes spectrum sharing among adjacent base stations, allowing them to cooperatively execute SIC through the exchange of decoded messages via feedback links. Conventional CoMP schemes in the literature are based on OMA \cite{Comp_OMA}, which have limited spectral efficiency due to the fact that once a resource block is occupied by a cell-edge user, it cannot be accessed by other users. Thus, the concept was extended in non-orthogonal downlink and uplink scenarios. Specifically, in \cite{DownlinkComp_PA1, DownlinkComp_PA2} the power allocation of a downlink CoMP system using NOMA was investigated. Furthermore, the performance of such a system in terms of outage probability and capacity was investigated in \cite{DownlinkComp_Analysis}. Both perfect and imperfect channel state information (CSI) were considered and insights on the effect of clustering order on the spectral efficiency of the network were extracted. The uplink scenario was studied in \cite{UplinkComp}. Adopting a stochastic geometry approach, closed-form expressions for the outage probability and the ergodic rate of the system were extracted and the superiority of NOMA over OMA was validated.   

The flexibility of CoMP is further enhanced by introducing a global medium access control (MAC) entity, an advancement made possible by the cloud-radio access network (C-RAN) technology \cite{C-RAN}. In this context, the authors in \cite{DU-NOMA-Pappi} proposed applying distributed NOMA (DU-NOMA) to the uplink of C-RANs, focusing on maximizing the capacity region through adaptive transmission rates. A key assumption in this work is that RRHs exchange information via high-capacity, error-free links. In \cite{DU-NOMA-Panos}, the performance of DU-NOMA for fixed transmission rates was analyzed and an optimal joint user association and decoding order selection scheme was proposed. In this work, the impact of outage in the feedback link on the performance of the system was also investigated. The concept of cooperation between RRHs was also extended for the case of rate-splitting multiple (RSMA) access in \cite{DU-RSMA-Panos}, where distributed uplink RSMA (DU-RSMA) was introduced. Closed-form expressions for the outage probability, ergodic rate, and a new metric termed as fill factor were extracted, while simulation results showcased the superiority over DU-NOMA, and the extension of the achievable capacity region. The positioning of  RRHs significantly affects coverage, interference, and throughput. To capture these effects, stochastic geometry offers a tractable mathematical framework for modeling and analyzing wireless networks with randomly deployed infrastructure nodes that is a more realistic representation of future wireless networks \cite{SG-Cellular-Andrews}.

Game theory has emerged as a robust mathematical framework for resolving complex optimization problems in wireless networks. For instance, a game theoretic approach was deployed in \cite{Games-MIMO} to address user-centric access point (AP) selection, aiming to maximize the sum spectral efficiency within cell-free massive multiple-input multiple-output (MIMO) systems. Beyond access and connectivity management, game theory provides solutions for dynamic resource allocation and computational task offloading, as demonstrated in vehicular edge computing (VEC) environments \cite{VEC}. As cooperation continues to emerge as a critical networking paradigm for enhancing performance across multiple layers, coalitional games prove to be a remarkably powerful mathematical tool for overcoming inherent design challenges and developing fair, robust, practical, and efficient cooperation strategies \cite{Coalition_Games}. Notable applications of coalitional games include enhancing physical layer security in wireless communications \cite{PLS} and optimizing resource allocation within heterogeneous unmanned aerial vehicle (UAV) networks \cite{Coalition-UAV}. Furthermore, the integration of game theoretic principles with generative artificial intelligence is anticipated to be a pivotal enabler for the architectural modeling and optimization of next-generation 6G wireless networks \cite{Games-GAI}.  

\subsection{Motivation and Contribution}

Based on the aforementioned technical literature review, the existing literature has not researched the concept of distributed cooperative massive access (DCMA) in large scale networks. Taking into account the promising results of DU-NOMA and DU-RSMA in small setups, as well as the fact that real-world networks are required to support massive connectivity and extremely reliable low latency communication (eRLLC), such an extension can prove vital both for academia and industry. This study provides not only useful insights into the improvements provided by the proposed framework, but also novel algorithms for efficient real-world implementation. In particular, the contributions of this work can be summarized as follows: 
\begin{itemize}
    \item We consider a DCMA framework operating on a partially centralized cell-free C-RAN architecture. The RRHs have decoding capabilities and cooperate with each other by exchanging decoded messages through dedicated feedback links. The proposed framework, utilizing stochastic geometry, extends the existing analysis, which focuses on small-scale network setups, as it considers a large number of randomly distributed users and RRHs that both follow a homogeneous Poisson point process (PPP). 
    \item We propose a novel synergetic decoding algorithm that resolves the assignment and message sharing routing for each user in the network. This algorithm provides an efficient way for the scalable and practical deployment of such networks, utilizing a user-centric method for the formulation of RRH clusters where every decoding and sharing operation occurs within them, regarding a specific user. This way it accounts for practical constraints, operating without the assumption of perfect global CSI and eliminating the need to exchange decoded messages across the entire network.  
    \item We apply a game theoretic approach to minimize the number of  utilized RRHs, without degrading the performance of the proposed synergetic decoding algorithm. Specifically, we formulate this optimization problem as a dynamic coalition formation game and develop a merge-and-split algorithm, achieving a significant reduction in the utilization of RRHs and feedback links. This improvement reduces the computational load, alleviates the fronthaul payload, and significantly simplifies the synchronization and routing task, while inherently reducing the total power consumed
    \item We compare the proposed system model with systems that do not consider cooperating RRHs or implement SIC during the decoding phase. Simulation results demonstrate the significant superiority of our framework in terms of outage probability. Collectively, these contributions pave the way for the scalable design and optimization of large-scale DCMA-enabled networks. Furthermore, we evaluate the performance of the proposed algorithms and validate their efficiency and expected behavior. Finally, the trade-off between outage performance, CSI acquisition overhead and feedback links utilization is highlighted, providing useful practical insights.
\end{itemize}


\subsection{Structure}
The rest of the paper is organized as follows: In Section II the considered system model is proposed. Section III presents an efficient algorithm for implementation under realistic practical constraints, while Section IV formulates and solves an optimization problem to minimize the number of required RRHs utilizing game theory. In Section V, simulation results are presented that highlight the superiority of the proposed framework over benchmarks and validate the efficiency of the proposed algorithms. Finally, Section VI concludes the paper.

\section{System Model}
We consider a wireless network that consists of users that operate on the same resource block and RRHs that serve these users. We focus on the uplink scenario and adopt a cell-free architecture where users are not permanently assigned to a specific RRH and can be potentially decoded by every RRH in the network. We utilize a partially centralized cloud architecture where the RRHs have decoding capabilities and the central unit is used to perform complex calculations and is responsible for the system's coordination. The fact that differentiates this system from other approaches is it's distributed and cooperative nature. The RRHs can exchange decoded messages through a dedicated feedback link which is assumed to be ideal and error-free. This can be implemented by using high capacity fiber optics. The RRHs are able to perform successive interference cancellation (SIC) and through the aforementioned cooperation mitigate interference across the entire network.

\begin{figure}[!h]
    \centering
    \includegraphics[scale = 0.25]{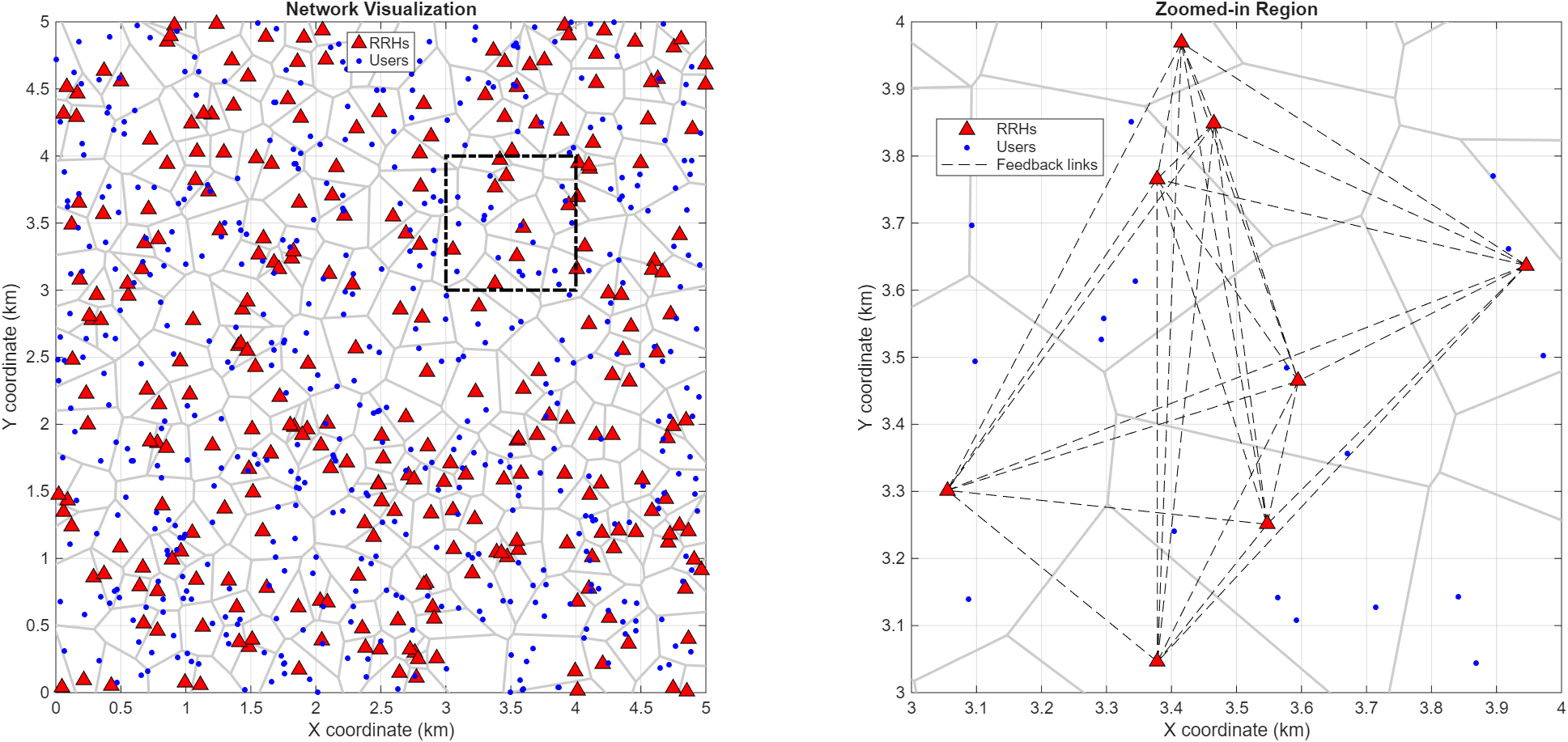}
    \caption{System model. In the left picture we see a realization of the network in a $5\times5$ km square with $\lambda_U=20$ users/$\text{km}^2$ and  $\lambda_{\mathrm{RRH}}=10$ RRHs/$\text{km}^2$. The right picture shows the zoomed-in region inside the black box highlighting the feedback links that connect the RRHs.}
    \label{System model}
\end{figure}

We denote as $U$ and $R$ the number of users and RRHs that exist in the network, respectively. We utilize stochastic geometry for the locations of both users and RRHs. Both are distributed according to a homogeneous Poisson point process (PPP) of intensity $\lambda_U$ and $\lambda_{\mathrm{RRH}}$, respectively. Consequently, $\mathbb{E}[ U]=\lambda_U\mathcal{A}$ is the expected number of users, while $\mathbb{E}[ R]=\lambda_{\mathrm{RRH}}\mathcal{A}$ is the expected number of RRHs in a network with area $\mathcal{A}$. The received signal-to-interference-plus-noise ratio (SINR) for a user $i$ at an RRH $j$ is expressed as 
\begin{equation}
    \gamma_{i,j}=\frac{\,h_{i,j}\,d_{i,j}^{-\alpha}}{I +N_0},
\end{equation}
where $d_{i,j}$ denotes the distance between user $i$ and RRH $j$, and
\begin{equation}
    I=\sum \limits_{k \neq i}\,h_{k,j}\,d^{-\alpha}_{k,j}
\end{equation}
is the cumulative interference from every other user on the network in RRH $j$ and $d_{k,j}$ denotes their distance to it.
Furthermore, $h_{n,j}$ with $n\in\{i,k\}$ denotes the small-scale fading power gain between the $j$-th RRH and the $n$-th user.
Assuming Rayleigh fading, $h_{n,j}$ follows the exponential distribution with rate parameter equal to 1. In addition, $N_0=\frac{\sigma^2}{P}$ is the effective noise power with $\sigma^2$ the power of additive white Gaussian noise (AWGN) and $P$ the transmit power of each user.
The standard power loss propagation model $d^{-\alpha}$ is used with path loss coefficient $\alpha>2$. It is known that in stochastic geometry the mean interference diverges and becomes infinite unless $\alpha$ is larger than the number of dimensions, which in our case is 2 \cite{random_distances}. The $\gamma_0=\frac{P}{\sigma^2}$ is defined to be the received signal-to-noise ratio (SNR) at a distance of $d=1$. 

All users transmit with the same fixed rate $R_i$ and the SINR threshold for successful decoding of the message of user $i$ is $\gamma_{\mathrm{th}}=2^{R_i}-1$. Furthermore, it is assumed that the channels are perfectly estimated by the RRHs while the users have no available channel state information (CSI).
The decoding process occurs in rounds which consist of a decoding and a message sharing phase. In each round, the system decodes every user that reaches the SINR threshold in at least one RRH of the network and thus can be successfully decoded. We can assume that the RRH that is assigned to decode a user is the one that achieves the highest received SINR. However, this assumption does not affect the performance evaluation of the system and the $i$-th user can be assigned to any $j$-th RRH where $\gamma_{i,j}\geq\gamma_{\mathrm{th}}$. After the decoding phase, the message sharing phase takes place, where the decoded messages are shared across the entire network and SIC is performed. Consequently, the interference caused by these messages is removed giving the opportunity for other users to be decoded in the next round. This iterative process continues until no other users can be decoded.

Let $\tau\in\{1,2,\dots,\tau_{\max}\}$ denote the round of decoding process with $\tau_{\max}$ being the last round. Consequently, $\gamma_{i,j}^{(\tau)}$ stands for the SINR of the user $i$ at RRH $j$ during round $\tau$. We introduce an indicator function 
\begin{equation}
    \mathbbm{1}_{\mathrm{out}}^{(i)} = 
\begin{cases} 
1, & \text{if } \gamma_{i,j}^{(\tau_{\max})} < \gamma_{\mathrm{th}} \quad \forall j\in \{1,2,\dots,R\} \\ 
0, & \text{if } \exists j \in \{1, 2, \dots, R\},\exists \tau:\gamma_{i,j}^{(\tau)} \geq \gamma_{\mathrm{th}} ,
\end{cases}
\end{equation}
which evaluates to 1 if the user $i$ remains undecoded at the end of the iterative process, and to 0 if it is successfully decoded during the process. We define the total system outage probability as
\begin{equation}\label{eq:outage-definition}
    P_{\mathrm{out}}=\frac{1}{U } \sum\limits_{i=1}^U\mathbbm{1}_{\mathrm{out}}^{(i)}.
\end{equation}
In the rest of the paper, when we state outage probability we refer to the total system outage probability as defined in \eqref{eq:outage-definition}. 


\section{Proposed Algorithm} \label{ProposedAlgorithm}
\subsection{Algorithm Formulation}
In this section, we present an efficient algorithm that implements the distributed decoding taking into account practical constraints such as perfect CSI for specific users and finite sharing links between the RRHs. We first present the parameters and variables that will be used in this algorithm.

\begin{table}[!h]
\vspace{-2mm}
\centering
\caption{List of Variables and Parameters}
\begin{tabular}{c||p{6.5cm}}
\hline
\textbf{Symbol} & \textbf{Description} \\ \hline \hline
$\mathcal{U}$ & Set of all users, where $u \in \{1, 2, \dots,U\}$. \\ \hline
$\mathcal{R}$ & Set of all RRHs, where $r \in \{1, 2, \dots,R\}$.  \\ \hline
$\beta_{u,r}$ & Path loss coefficient between user $u$ and RRH $r$. \\ \hline
$s_{u,r}$ & Instantaneous received signal power from user $u$ at RRH $r$. \\ \hline
$P_{\mathrm{tot},r}$ & Total instantaneous received power at RRH $r$ from all the users in the network. \\  \hline
$\beta_{\mathrm{th}}$ & Threshold for the establishment of initial clusters. It represents power and will be compared with $\beta_{u,r}$. \\ \hline
$P_{\mathrm{th}}$ & Threshold for sharing decoding messages. It is also a power threshold and will be compared with $s_{u,r}$. \\ \hline
$K_{\max}$ & Maximum number of RRHs allowed in a user's cluster. \\ \hline
$M_{\max}$ & Maximum checks per round for decoding a user. \\  \hline
$\mathcal{C}_u$ & The established RRH cluster for user $u$. \\ \hline
$\mathcal{L}_u$ & Candidate RRH checklist for decoding user $u$. \\ \hline
$\mathcal{V}_u$ & Set of valid RRHs to enter the user's $u$ cluster. \\ \hline
$\mathcal{U}_{\mathrm{a}}$ & The set of users that have not yet been decoded. \\ \hline
$\mathcal{R}_{\mathrm{a}}$ & The "Active Listener" set, representing RRHs that are candidates for decoding from the active users in the system. \\ \hline
$\mathcal{R}_{\mathrm{interf}}$ & The subset of RRHs that belong to the user's $u$ cluster, $\mathcal{C}_u$, and his instantaneous power exceeds the sharing threshold $s_{u,r}\geq P_{\mathrm{th}}$. \\ \hline
$\mathcal{U}_{\mathrm{dec}}$ & The set of successfully decoded users. \\ \hline
$\mathcal{U}_{\mathrm{new,dec}}$ & The set of users that were successfully decoded in the current round. \\ \hline
$r^*_u$ & The RRH that is assigned to decode user $u$. \\ \hline
$\mathcal{R}_{\mathrm{s},u}$ & The set of RRHs that need the decoded message of user $u$ to execute SIC,  where $\mathcal{R}_{\mathrm{s},u} \subseteq \mathcal{C}_u$.  \\ \hline
\end{tabular}
\label{tab:algorithm_variables}
\end{table}

\begin{algorithm}[!ht]
\caption{Efficient DCMA with Active Listener Sharing}    
\label{alg:active_listener}
\begin{algorithmic}[1] 
\Require $\mathcal{U}, \mathcal{R}, \boldsymbol{\beta}, \mathbf{s}, \boldsymbol{P}_{\mathrm{tot}},\gamma_{\mathrm{th}}, \beta_{\mathrm{th}}, P_{\mathrm{th}}, K_{\max}, M_{\max}, N_0$
\Ensure $\mathcal{U}_{\mathrm{dec}}$, $r^*_u$, $\mathcal{R}_{\mathrm{s},u}$
\Statex
\State Initialize: $\mathcal{U}_{\mathrm{a}} \gets \mathcal{U}$, $\mathcal{U}_{\mathrm{dec}} \gets \emptyset$
\Statex \textbf{Stage 1: RRH Clustering and CSI Acquisition}
\For{each user $u \in \mathcal{U}$}
    \State $\mathcal{V}_u \gets \{r \in \mathcal{R} \mid \beta_{u,r} \geq \beta_{\mathrm{th}} \}$ 
    \If{$|\mathcal{V}_u| > K_{\max}$}
        \State Sort $\mathcal{V}_u$ in descending order based on $\beta_{u,r}$
        \State $\mathcal{C}_u \gets$ Top $K_{\max}$ elements of the sorted $\mathcal{V}_u$
    \Else
        \State $\mathcal{C}_u \gets \mathcal{V}_u$
    \EndIf
    \State Obtain $s_{u,r}$ for $r \in \mathcal{C}_u$
    \State Sort $\mathcal{C}_u$ in descending order based on $s_{u,r}$
    \State $\mathcal{L}_u \gets$ Top $\min{(|\mathcal{C}_u|, M_{\max})}$ elements of the sorted $\mathcal{C}_u$
\EndFor
\Statex \textbf{Stage 2: Decoding and Message sharing}
\State $\mathcal{U}_{\mathrm{new,dec}}\gets \mathcal{U}$ \Comment{Initialize $\mathcal{U}_{\mathrm{new,dec}} \neq \emptyset$ in order to enter the loop}
\While{$\mathcal{U}_{\mathrm{new,dec}} \neq \emptyset$}
    \Statex \textbf{\quad Phase 1: Decoding}
    \State $\mathcal{U}_{\mathrm{new,dec}} \gets \emptyset$
    \For{each active user $u \in \mathcal{U}_{\mathrm{a}}$}
        \For{$i=1$ to $|\mathcal{L}_u|$}
            \State $r \gets \mathcal{L}_u[i]$
            \State $I_{u,r} \gets P_{\mathrm{tot},r} - s_{u,r}$
            \State $\gamma_{u,r} \gets \frac{s_{u,r}}{I_{u,r}+N_0}$
            \If{$\gamma_{u,r}\geq\gamma_{\mathrm{th}}$}
                \State $\mathcal{U}_{\mathrm{dec}} \gets \mathcal{U}_{\mathrm{dec}} \cup \{u\}$
                \State $\mathcal{U}_{\mathrm{new,dec}} \gets \mathcal{U}_{\mathrm{new,dec}} \cup \{u\}$
                \State $r^*_u \gets r$
                \State \textbf{break}
            \EndIf 
        \EndFor
    \EndFor
    \Statex \textbf{\quad Phase 2: Message sharing}
    \If{$\mathcal{U}_{\mathrm{new,dec}} \neq \emptyset$}
        \State $\mathcal{U}_{\mathrm{a}} \gets \mathcal{U}_{\mathrm{a}} \setminus \mathcal{U}_{\mathrm{new,dec}}$
        \State $\mathcal{R}_{\mathrm{a}} \gets \bigcup_{u\in \mathcal{U}_{\mathrm{a}}}\mathcal{L}_u$
        \For{each user $u \in \mathcal{U}_{\mathrm{new,dec}}$}
            \State $\mathcal{R}_{\mathrm{interf}} \gets \{r \in \mathcal{C}_u \mid s_{u,r} \geq P_{\mathrm{th}} \}$
            \State $\mathcal{R}_{\mathrm{s},u}\gets \mathcal{R}_{\mathrm{interf}} \cap \mathcal{R}_{\mathrm{a}}$
            \For{each $r \in \mathcal{R}_{\mathrm{s},u}$}
                \State $P_{\mathrm{tot},r}\gets P_{\mathrm{tot},r}-s_{u,r}$
            \EndFor
        \EndFor
    \EndIf 
\EndWhile
\end{algorithmic}
\end{algorithm}

In Stage 1 of the algorithm, each user obtains his own cluster that includes RRHs that are possible candidates for decoding him and also those that will need the decoded message to enhance their decoding capabilities during SIC. This is done using a user-centric method. We use as a criterion the path loss coefficient $\beta_{u,r}=d_{u,r}^{-\alpha}$ where $d_{u,r}$ is the distance between user $u$ and RRH $r$. Specifically, when $\beta_{u,r}\geq \beta_{\mathrm{th}}$, RRH $r$ is included in the user's cluster. The parameter $K_{\max}$ denotes the maximum number of RRHs that the user's cluster is allowed to have. It is introduced to account for limitations in CSI overhead and feedback link utilization and cap the cluster size to a certain number because it directly affects them. Taking this into account, if more than $K_{\max}$ RRHs fulfill the above mentioned requirement, we sort them in descending order according to $\beta_{u,r}$ and keep the first $K_{\max}$ RRHs resulting in the final cluster $\mathcal{C}_u$. We should notice that this process is based on the path loss coefficient that can be estimated via large-scale fading estimation, which is a lower complexity process than the small-scaling fading estimation where the complex channel gain is needed. Thus, in the proposed algorithm, small-scale fading estimation is required only for the RRHs that are included in a user's cluster, reducing the implementation complexity. This means that from now on we consider $s_{u,r}=h_{u,r}\,d_{u,r}^{-\alpha}$ to be known for each RRH inside the cluster. After the CSI is obtained, we characterize the top $M_{\max}$ RRHs with the largest  $s_{u,r}$ as those that will be checked to decode that user and obtain the $\mathcal{L}_u$ set. If the number of available RRHs in the cluster is insufficient to meet the limit $M_{\max}$, all RRHs within the cluster are allocated to the set $\mathcal{L}_u$. Consequently, the total number of RRHs in $\mathcal{L}_u$ is given by $\min(|\mathcal{C}_u|, M_{\max})$.

Stage 2 is an iterative process in which we check for each active user $u$ if there is at least one of the RRHs in $\mathcal{L}_u$ that can decode his message. This process is done starting from the RRH with the largest $s_{u,r}$ and we continue in descending order. When we find a valid RRH, the process stops and we move into the next user. Each decoding round ends when there are no other users that can be decoded in any RRH. This concludes the decoding phase of the algorithm. Subsequently, in accordance with the aforementioned decoding process, the message sharing phase begins, where the decoded messages are shared with neighboring RRHs. Specifically, for each user that was decoded in this round, the RRHs to which the decoded message should be transmitted must be selected. These are decided by taking two criteria into account. The first criterion for an RRH $r$ to receive the decoded message of user $u$ is that the RRH must be included in the user's cluster $\mathcal{C}_u$ and $s_{u,r}\geq P_{\mathrm{th}}$. If this is true, then the user's interference on that RRH is non-negligible, and its cancellation is needed for the process to continue.

For the second criterion, the "Active Listener" logic is applied, where we track which RRHs are possible decoding candidates for the active users in the system. If an RRH is not part of the $\mathcal{R}_{\mathrm{a}}$ set, even if it receives the decoded message, it will not facilitate the system to decode more users in the future. So, the RRHs that fulfill those two criteria create the set $\mathcal{R}_{\mathrm{s},u}$ where the RRH $r^*_u$ that decoded this specific user will pass the information to the RRHs that belong to $\mathcal{R}_{\mathrm{s},u}\setminus\{r^*_u\}$. Then the interference is subtracted and the next round starts. This iterative process continues until no other user can be decoded.

\subsection{Computational Complexity}
In this subsection, a comparison between the complexity of the proposed algorithm and the complexity of an exhaustive search is presented. As stated earlier, we assume that there are $U$ users and $R$ RRHs in the system. Regarding the proposed algorithm, the creation of the clusters takes place first. This requires $\beta_{u,r} \geq \beta_{\mathrm{th}}, \forall u\in\mathcal{U},\forall r\in\mathcal{R}$ to be checked, leading to $\mathcal{O}(U\, R)$ complexity. Considering the worst case scenario in which every user has $R$ RRHs that fulfill this condition,  the sorting algorithm for each user needs $R \log R$ operations. Consequently, the total complexity for all users is $\mathcal{O}(U\, R \log R)$. By adding these two, we can obtain the cluster formulation complexity as
\begin{equation}
    \mathcal{O}(U\, R+U\, R \log R)=\mathcal{O}(U\, R \log R)
\end{equation}
by keeping the dominant term. Then, having the channel estimations, we sort them. Each user has a maximum $K_{\max}$ RRHs in his cluster. This translates into $K_{\max} \log K_{\max}$ operations per user. So, the total complexity is $\mathcal{O}(U\, K_{\max} \log K_{\max})$.

For the iterative decoding process, the worst-case scenario is considered, where only one user is decoded per round. This means that in the first round we check $U$ users, in the second round $U-1$ users, while in the final round only one user is checked. This forms an arithmetic progression, which enables us to calculate the total required checks as
\begin{equation}
    U+(U-1)+(U-2)+\dots+1=\frac{U(U+1)}{2}=\frac{U^2}{2}+\frac{U}{2}.
\end{equation}
At each check, the algorithm considers at most $M_{\max}$ RRHs. So, the total complexity becomes $\mathcal{O}(U^2\, M_{\max})$ keeping only the dominant term. At each iteration, we build a mask size $R$ to track which RRHs are still needed by the remaining undecoded users. This is implemented by a Boolean array, and to check if an RRH is active inside it requires $\mathcal{O}(1)$. It maps all the $M_{\max}$ indices for all active users. This has the same form as the aforementioned arithmetic progression, which means that it has complexity $\mathcal{O}(U^2\, M_{\max})$. Afterwards, the first criterion of message sharing needs to be checked, i.e., which RRHs have $s_{u,r} \geq P_{\mathrm{th}}$, with complexity $\mathcal{O}(U\, K_{\max})$. Finally, the RRHs that are both part of the mask and fulfill the above criterion will receive the decoded message, which can be checked with complexity $\mathcal{O}(U\, K_{\max})$. It should be noted that, the interference is subtracted at most $K_{\max}$ RRHs, which means $\mathcal{O}(U\, K_{\max})$. By adding all the complexity calculations after the cluster formulation, we get
\begin{equation}
    \mathcal{O}(U\, K_{\max} \log K_{\max}+2\, U^2\, M_{\max}+3\, U\, K_{\max} )=\mathcal{O}(U^2)
\end{equation}
by keeping only the dominant terms and removing the constant ceiling parameters $K_{\max}$ and $M_{\max}$ because they are not affected by the network scaling variables $U$ and $R$. 
This results in the final complexity of the algorithm 
\begin{equation}
    \mathcal{O}(U\, R \log R+U^2).
\end{equation}

The ideal exhaustive algorithm where the user is always decoded in the RRH with the highest SINR and his decoded message is shared across the entire network, has complexity $\mathcal{O}(U^2\, R)$. This can be proved if we consider the worst-case scenario again. The difference here is that we must check all the $R$ RRHs of the network and not the $M_{\max}$ that we did before. This yields a complexity of $\mathcal{O}(U^2\, R)$. After that, interference cancellation occurs in each RRH, which is of complexity $\mathcal{O}(U\, R)$. By keeping the dominant term we end up with the final complexity. 

It is apparent that the proposed algorithm is more efficient in terms of complexity. Even if we make a numerical comparison between $\mathcal{O}(U\, R \log R+U^2)$ and $\mathcal{O}(U^2\, R)$ we can see a significant gain. It should also be highlighted that the cluster formulation phase that gives the $\mathcal{O}(U\, R \log R)$ term does not need to be executed after the completion of each round of the decoding process. It depends entirely on the slow-varying path loss coefficient, thus the cluster reformulation is needed when we have significant changes in the users locations. This means that in the decoding process of the algorithm, where the algorithm must execute continuously across each coherence interval, we see that our proposed algorithm has a complexity of $\mathcal{O}(U^2)$ in contrast to $\mathcal{O}(U^2\, R)$ that the exhaustive algorithm has. This decouples the complexity from the number of RRHs in the network, drastically reducing the operations needed, and consequently allowing for scalability in massive deployments. In addition, it should be noted that it solves the exhaustive algorithm's problem of requiring the instantaneous CSI across the entire network. In the proposed algorithm, knowledge of the complex channel is required only inside the cluster because the operations that involve decoding and SIC occur only there.

It should also be noted that the derived complexity and the required operations are an upper bound since they are derived by taking the worst-case scenario into account. In a practical network, multiple users are decoded per round and each user is rarely decoded by RRHs with lower $s_{u,r}$, thus checking all $M_{\max}$ of the cluster is usually not required. 

\section{Network Optimization}
\subsection{Optimization Problem}
In this section, we aim to further optimize the operation of the network. Specifically, the utilization of fewer active RRHs provides significant structural gains for the proposed system. Primarily, it reduces the computational load at the central unit while simultaneously decreasing the number of required active feedback links. This directly alleviates the fronthaul payload and significantly simplifies the synchronization and routing task of the decoded message exchange. Furthermore, minimizing the active hardware footprint inherently reduces the total power consumed, rendering the overall network highly energy-efficient.

Taking these into account, the goal is to minimize the RRHs used, without degrading the system's performance. Therefore, for the established user clusters $\mathcal{C}_u$, which contain a maximum of $K_{\max}$ RRHs, our aim is to decode as many users as possible while simultaneously using the minimum number of RRHs that can achieve that. The best result with formulated clusters happens when we do not constrain the number of RRHs in which each user is checked to be decoded by and we exchange the decoded messages to every member in the cluster using the Active Listener technique. This is translated into $P_{\mathrm{th}}=0$ and $M_{\max}=K_{\max}$ using the parameters of the proposed algorithm. We introduce the set $\mathcal{U}_{\mathrm{t}}$ that denotes the users the network is able to decode under these conditions.



This can be formulated as
\begin{equation} \tag{\textbf{P1}}
    \label{optim_problem}
    \begin{aligned}
         \mathop{\mathrm{min}}\limits_{\boldsymbol{x}, \boldsymbol{z}, \boldsymbol{c}, \boldsymbol{l}} \quad & \sum\limits_{r \in \mathcal{R}} x_r  \\
         \textbf{s.t.} \quad
         & \mathrm{C}_1: \, \sum\limits_{r \in \mathcal{C}_u} z_{u,r} = 1, \qquad \forall u \in \mathcal{U}_{\mathrm{t}} \\
         & \mathrm{C}_2: \, z_{u,r} \le x_r, \quad \forall u \in \mathcal{U}_{\mathrm{t}}, \forall r \in \mathcal{C}_u \\
         & \mathrm{C}_3: \, c_{w,u}(l_w + 1 - l_u) \le 0, \quad \forall u, w \in \mathcal{U}_{\mathrm{t}}, u \neq w \\
         & \mathrm{C}_4: \, z_{u,r} \Biggl[ \gamma_{\mathrm{th}} \Biggl(\omega_r + \sum\limits_{w \in \mathcal{U}_{\mathrm{t}} \setminus \{u\}} s_{w,r} (1 - c_{w,u} a_{w,r}) \Biggr) \\
         & \qquad \qquad - s_{u,r} \Biggr] \le 0, \quad \forall u \in \mathcal{U}_{\mathrm{t}}, \forall r \in \mathcal{C}_u \\
         & \mathrm{C}_5: \, x_r, z_{u,r}, c_{w,u} \in \{0, 1\}, \quad l_u \in \{1, \dots, |\mathcal{U}_{\mathrm{t}}|\}.
    \end{aligned}
\end{equation}
Here, $x_r$ denotes the status of the $r$-th RRH with 1 being active, $z_{u,r}$ denotes the assignment status with 1 meaning that user $u$ is decoded by RRH $r$, $c_{w,u}$ denotes whether the message of user $w$ is decoded and subtracted before decoding user $u$. Furthermore, $l_u$ denotes the order in which user $u$ is decoded, $\omega_r$ denotes the noise power plus the interference floor at RRH $r$ caused by users in the network, whose messages cannot be shared with this RRH, and $a_{w,r}$ denotes whether the decoded message of user $w$ is shared with RRH $r$. To formulate the problem accurately, the system's physical and operational limits must be strictly defined. First, $\mathrm{C}_1$ enforces the strict service requirement of the network, guaranteeing that every user in $\mathcal{U}_{\mathrm{t}}$ is assigned to exactly one RRH within their cluster $\mathcal{C}_u$ and successfully decoded by it. A physical limitation between user decoding and hardware activation is captured by $\mathrm{C}_2$, ensuring a user is only decoded by a powered-on RRH. To enable SIC, $\mathrm{C}_3$ establishes a strict chronological order that prevents cyclic dependencies such as users $w$ and $u$ requiring each other to be canceled first. Having this constraint guaranties that subtracting user $w$'s power from the SINR equation for user $u$ requires $w$ to be decoded at a strictly earlier stage of the process. Furthermore, if user $u$ is decoded in RRH $r$, $z_{u,r}=1$, the necessary condition $\gamma_{u,r}\geq\gamma_{\mathrm{th}}$ is guaranteed by $\mathrm{C}_4$. Finally, $\mathrm{C}_5$ defines the variable domains of the problem.

Solving this optimization problem in \eqref{optim_problem} to global optimality is strongly NP-Hard and can be shown by decomposing it into two tightly coupled combinatorial subproblems, each independently belonging to the NP-Hard complexity class. Each subproblem comes from applying a relaxation in the constraints and creating a reduced form of the original problem. 

The first subproblem is given by 
\begin{equation} \tag{\textbf{P2}}
    \label{optim_problem_sub1}
    \begin{aligned}
         \mathop{\mathrm{min}}\limits_{\boldsymbol{x}, \boldsymbol{z}} \quad & \sum\limits_{r \in \mathcal{R}} x_r  \\ 
         \textbf{s.t.} \quad
         & \mathrm{C}_1: \, \sum\limits_{r \in \mathcal{C}_u} z_{u,r} = 1, \quad \forall u \in \mathcal{U}_{\mathrm{t}} \\
         & \mathrm{C}_2: \, z_{u,r} \le x_r, \quad \forall u \in \mathcal{U}_{\mathrm{t}}, \forall r \in \mathcal{C}_u \\
         & \mathrm{C}_5: \, x_r, z_{u,r} \in \{0, 1\},
    \end{aligned}
\end{equation}
where we assume that $c_{w,u},l_u$ are fixed and ignore constraint $\mathrm{C}_4$. The residual objective collapses to minimizing the hardware sum $\sum x_r$, constrained by $\mathrm{C}_1$ and $\mathrm{C}_2$. This reduced formulation is isomorphic to the classic minimum set cover (MSC) problem. The universe is the $\mathcal{U}_{\mathrm{t}}$ and the available covering subsets are dictated by the clusters $\mathcal{C}_u$. The MSC problem is proven to be NP-Hard and it has a discrete search space of $\mathcal{O}(2^{|\mathcal{R}|})$ combinations \cite{MSC}.


For the second subproblem, we consider $x_r$ and $z_{u,r}$ to be fixed. We must derive $c_{w,u}$ and $l_u$ to satisfy constraint $\mathrm{C}_4$. This subproblem generalizes the classic linear ordering problem (LOP) into a non-linear feasibility domain and is formulated as
\begin{equation} \tag{\textbf{P3}}
    \label{optim_problem_sub2}
    \begin{aligned}
         \mathop{\mathrm{find}}\limits_{\boldsymbol{c}, \boldsymbol{l}} \quad & \boldsymbol{c}, \boldsymbol{l} \\
         \textbf{s.t.} \quad
         & \mathrm{C}_3: \, c_{w,u}(l_w + 1 - l_u) \le 0, \quad \forall u, w \in \mathcal{U}_{\mathrm{t}}, u \neq w \\
         & \mathrm{C}_4: \, z_{u,r} \Biggl[ \gamma_{\mathrm{th}} \Biggl( \omega_r + \sum\limits_{w \in \mathcal{U}_{\mathrm{t}} \setminus \{u\}} s_{w,r} (1 - c_{w,u} a_{w,r}) \Biggr) \\
         & \qquad \qquad - s_{u,r} \Biggr] \le 0, \quad \forall u \in \mathcal{U}_{\mathrm{t}}, \forall r \in \mathcal{C}_u \\
         & \mathrm{C}_5: \, c_{w,u} \in \{0, 1\}, \quad l_u \in \{1, \dots, |\mathcal{U}_{\mathrm{t}}|\}.
    \end{aligned}
\end{equation}
In this reduction, each user $u\in \mathcal{U}_{\mathrm{t}}$ is a vertex in a directed graph. A directed edge from vertex $w$ to vertex $u$ exists if user $w$'s interference must be subtracted to successfully decode user $u$ ($c_{w,u}=1$). The constraint $\mathrm{C}_3$ is enforcing an acyclic sequence and creates dependencies, making this graph a strictly directed acyclic graph (DAG). The goal is to find a chronological sequence of users to mathematically validate the SINR condition for every user. Finding a feasible linear ordering is also an NP-Hard problem that has a temporal search space bounded by $\mathcal{O}(|\mathcal{U}_{\mathrm{t}}|!)$ combinations \cite{LOP}.

These problems are deeply interconnected via the logical implications mapped in $\mathrm{C}_3$ and $\mathrm{C}_4$. The global problem cannot be decomposed and solved in disjoint phases. Furthermore, the feasible region of the solution space is non-convex due to the integer nature of the variables. Consequently, an exact global optimizer is forced to evaluate the joint combinatorial space of $\mathcal{O}(2^{|\mathcal{R}|}|\mathcal{U}_{\mathrm{t}}|!)$.

\subsection{Game Theoretic Approach}
For the aforementioned reasons, we utilize game theory to solve \eqref{optim_problem}. Specifically, we formulate our optimization problem as a dynamic coalition formation game by the pair $(\mathcal{N},v)$ \cite{Coalition_Games}. The set of  players $\mathcal{N}=\mathcal{R}$ is the set of all RRHs available in the network. The function $v$ quantifies the worth of a coalition $\mathcal{S}\subseteq \mathcal{N}$ in the game, formally referred to as the utility. In addition, to better characterize the type of the game, it is classified as a cooperative game because the RRHs act as cooperative players seeking to form coalitions to act as a single functional entity and strengthen their decoding capabilities.

Since the objective is to minimize the number of RRHs used while achieving the decoding of the $\mathcal{U}_{\mathrm{t}}$ user set, the utility of a coalition is given by 
\begin{equation}
    v(\mathcal{S})= (\Phi(\mathcal{S}),-|\mathcal{S}|).
\end{equation}
As seen, the utility is a vector, whose first element, $\Phi(\mathcal{S})$, is a network evaluator function that returns the number of users decoded by the coalition $\mathcal{S}$ and the second element, $-|\mathcal{S}|$, represents the negative cardinality of the coalition.  

\begin{lemma} \label{Lex}
    Let $y=(y_1,y_2,\dots,y_k)$ and $z=(z_1,z_2,\dots,z_k)$ be two vectors in $\mathbb{R}^k$. We say that vector $y$ is strictly lexicographically greater than vector $z$ if and only if there exists some index $m\in\{1,\dots,k\}$ such that the elements are strictly equal up to $m$ and the $m$-th element of $y$ is strictly greater than the $m$-th element of $z$. This is written as
\begin{equation}
    y\succ_{lex}z \iff \exists m \le k:\forall (i< m,y_i=z_i) \land (y_m>z_m).
\end{equation}
\end{lemma}

To evaluate and compare the worth of two coalitions $\mathcal{S}_1$ and $\mathcal{S}_2$ fairly, a lexicographic preference is adopted in our game. Specifically, taking Lemma \ref{Lex} into account, it is said that the utility of coalition $\mathcal{S}_1$ is greater than that of $\mathcal{S}_2$ when 
\begin{equation}
    \begin{split}
        v(\mathcal{S}_1) \succ_{lex} v(\mathcal{S}_2) \iff & \Big(\Phi(\mathcal{S}_1) > \Phi(\mathcal{S}_2)\Big) \\
        &\!\! \lor \Big(\!\Phi(\mathcal{S}_1)\! =\! \Phi(\mathcal{S}_2) \!\land\! -|\mathcal{S}_1|\! >\! -|\mathcal{S}_2| \Big).
    \end{split}
\end{equation}
By structuring the utility lexicographically, the framework mathematically guaranties that redundant RRHs will be powered off, maximizing the second element of the utility, while ensuring that the primary constraint of user coverage is never compromised.

In canonical games, cooperation is assumed to be universally beneficial, meaning that the optimal state is the grand coalition where every player joins in a common coalition. In our case, the game does not have the superadditivity property because each active RRH incurs a hardware penalty $x_r$ and arbitrarily adding RRHs to the coalition increases cost without necessarily improving coverage. Therefore, the grand coalition is almost never the optimal topology. 

\begin{lemma} \label{max_lex}
    Given a finite, non-empty set of vectors $\mathcal{X} \subset \mathbb{R}^k$, there exists a maximal vector $x^* \in \mathcal{X}$ under the lexicographic order. The lexicographic maximum operator, denoted as $\max_{lex}$, isolates this vector such that
    \begin{equation}
        x^*= \max_{lex} \mathcal{X} \implies x^* \succeq_{lex} x, \quad \forall x\in \mathcal{X},
    \end{equation}
    where $\succeq_{lex}$ denotes the weak lexicographic preference relation that holds if and only if $x^* \succ_{lex} x$ or $x^* = x$.
\end{lemma}
\begin{IEEEproof}
    Since $\mathcal{X}$ is a finite set and the lexicographic preference relation establishes a total order over $\mathbb{R}^k$, the vectors in $\mathcal{X}$ can be strictly sorted. Consequently, a maximal element with respect to this total order is mathematically guaranteed to exist within the set.
\end{IEEEproof}

A coalition structure $\pi=\{\mathcal{S}_1,\mathcal{S}_2,\dots,\mathcal{S}_K\}$ with $K\le|\mathcal{N}|$ is a partition of $\mathcal{N}$. For each coalition in the structure, it holds that $\mathcal{S}_k\neq\emptyset$, $\bigcup_{k=1}^K \mathcal{S}_k = \mathcal{N}$, and $\mathcal{S}_k \cap\mathcal{S}_l = \emptyset, \forall k,l\in\{1,2,\dots,K\}$ with $k\neq l$. In our case, the network partition will have an active coalition $\mathcal{R}_{\mathrm{on}}\in\pi$ where the powered-on RRHs that collaborate to serve the users are included. In addition, the deactivated RRHs can be mathematically treated as singleton coalitions. The utility of a deactivated singleton coalition is assumed to be $(0,-1)$. We define the total social welfare of a coalition structure $\pi$ as 
\begin{equation}
    W(\pi) = \max_{lex}{\{v(\mathcal{S}) \mid \mathcal{S} \in \pi\}}.
\end{equation}

\begin{algorithm}[!hb]
\caption{Stochastic Lexicographic Merge-and-Split} 
\label{alg:slms}
\begin{algorithmic}[1] 
\Require $\mathcal{R}$, $\mathcal{U}_{\mathrm{t}}$, $\pi_{\mathrm{seed}}$, $\mathcal{R}_{\mathrm{crit}}$ 
\Ensure $\pi_{\mathrm{opt}}$
\State Initialize partition state: $\pi \gets \pi_{\mathrm{seed}}$
\State Initialize $K_{\mathrm{limit}} \gets |\mathcal{R}_{\mathrm{on}}|$
\While{\textbf{true}}
    \Statex \quad \textbf{Phase 1: Utilitarian Split}
    \State $\Pi_{\mathrm{valid}} \gets \{\pi_{-r} \mid r \in \mathcal{R}_{\mathrm{on}} \setminus \mathcal{R}_{\mathrm{crit}} \land W(\pi_{-r}) \succ_{lex} W(\pi)\}$
    \If{$|\Pi_{\mathrm{valid}}| > 0$}
        \State $\pi^* \in_R \Pi_{\mathrm{valid}}$
        \State $\pi \gets \pi^*$
        \State \textbf{continue} \Comment{Restart loop to evaluate further splits}
    \EndIf
    \Statex \quad \textbf{Phase 2: Target Evaluation}
    \If{$\Phi(\mathcal{R}_{\mathrm{on}}) = |\mathcal{U}_{\mathrm{t}}|$}
        \State \textbf{break} \Comment{Target met and footprint is fully pruned}
    \EndIf
    \Statex \quad \textbf{Phase 3: State Space Escalation \& Utilitarian Merge}
    \If{$|\mathcal{R}_{\mathrm{on}}| = K_{\mathrm{limit}}$} 
        \State $K_{\mathrm{limit}} \gets K_{\mathrm{limit}} + 1$ 
    \EndIf
    \State Evaluate: $W_r \gets W(\pi_{+r}), \forall r \in \mathcal{R} \setminus \mathcal{R}_{\mathrm{on}}$ 
    \State Find max: $W_{\max} \gets \max_{lex}(W_r)$
    \If{$W_{\max} \succ_{lex} W(\pi)$}
        \State $\Pi_{\mathrm{best}} \gets \{\pi_{+r} \mid r \in \mathcal{R} \setminus \mathcal{R}_{\mathrm{on}} \land W(\pi_{+r}) = W_{\max}\}$
        \State $\pi^* \in_R \Pi_{\mathrm{best}}$
        \State $\pi \gets \pi^*$
    \Else
        \State $\pi \gets \pi_{\mathcal{R}}$ 
        \State \textbf{break} 
    \EndIf
\EndWhile
\State \Return $\pi_{\mathrm{opt}} \gets \pi$
\end{algorithmic}
\end{algorithm}

To solve this game, we have developed a merge-and-split algorithm that follows the utilitarian order. According to this order, the players prefer to alter the partition $\pi$ into $\pi^\prime$ if this new topology increases the total social welfare, i.e., $W(\pi')\succ_{lex}W(\pi)$. As we can see, there are the partition $\pi_{\mathrm{seed}}$ and the RRH set $\mathcal{R}_{\mathrm{crit}}$ that are the initialization of the algorithm and a set of mandatory RRHs that cannot be removed from the active coalition, respectively. These two are explained in detail later. We define $\pi_{-r}$, the new partition created by splitting the RRH $r$ from the active coalition into a singleton one and powering it off. Similarly, $\pi_{+r}$ denotes the new coalition structure created by merging the singleton coalition that includes RRH $r$ with the active coalition. The variable $K_{\mathrm{limit}}$ bounds the maximum number of RRHs that can exist in the active coalition in the current stage of the algorithm. Furthermore, $\pi$ is the current partition state, $\pi^*$ is the new partition formed after a merge or split operation, and $\pi_{\mathcal{R}}$ denotes the partition that includes the grand coalition where $\mathcal{R}_{\mathrm{on}}=\mathcal{R}$.

\begin{algorithm}[!b]
\caption{Mandatory RRHs and Initialization} 
\label{alg:seed_init}
\begin{algorithmic}[1] 
\Require $\mathcal{R}, \mathcal{U}_{\mathrm{t}}, \mathcal{C}_u, \textbf{s}, N_0, \gamma_{\mathrm{th}}, \boldsymbol{P}_{\mathrm{tot}}$ 
\Ensure $ \pi_{\mathrm{seed}}, \mathcal{R}_{\mathrm{crit}}$
\Statex \quad \textbf{Phase 1: Feasibility Pruning}
\State Compute residual interference plus noise $\forall r \in \mathcal{R}$: $\omega_r \gets N_0 + P_{\mathrm{tot},r} - \sum_{\{v \in \mathcal{U}_{\mathrm{t}} \mid r \in \mathcal{C}_v\}} s_{v,r}$
\State Prune infeasible links $\forall u \in \mathcal{U}_{\mathrm{t}}$: $\tilde{\mathcal{C}}_u \gets \left\{ r \in \mathcal{C}_u \;\middle|\; \frac{s_{u,r}}{\omega_r } \geq \gamma_{\mathrm{th}} \right\}$
\Statex \quad \textbf{Phase 2: Mandatory RRHs Extraction}
\State $\mathcal{R}_{\mathrm{crit}} \gets \bigcup_{\{u \in \mathcal{U}_{\mathrm{t}} \mid |\tilde{\mathcal{C}}_u| = 1\}} \tilde{\mathcal{C}}_u$ 

\Statex \quad \textbf{Phase 3: Seed Selection}
\If{$|\mathcal{R}_{\mathrm{crit}}| > 0$}
    \State $\mathcal{R}_{\mathrm{on}} \gets \mathcal{R}_{\mathrm{crit}}$
\Else
    \State $\mathcal{R}_{\mathrm{useful}} \gets \bigcup_{u \in \mathcal{U}_{\mathrm{t}}} \mathcal{C}_u$
    \State $r^* \in_R \mathcal{R}_{\mathrm{useful}}$, \quad $\mathcal{R}_{\mathrm{on}} \gets \{r^*\}$ 
\EndIf
\State Form initial partition: $\pi_{\mathrm{seed}} \gets \{\mathcal{R}_{\mathrm{on}}\} \cup \{\{r\} \mid r \in \mathcal{R} \setminus \mathcal{R}_{\mathrm{on}} \}$
\end{algorithmic}
\end{algorithm}

The algorithm aiming to remove redundant RRHs evaluates the temporary isolation of each non-critical active RRH, $\mathcal{R}_{\mathrm{on}} \setminus \mathcal{R}_{\mathrm{crit}}$, into a singleton. This split is considered valid if the resulting partition, $\pi_{-r}$, yields a total social welfare $W(\pi_{-r})$ that is lexicographically greater than the current baseline $W(\pi)$. Since removing an RRH inherently improves the secondary objective, this condition mathematically guaranties that the drop is accepted if and only if the primary objective remains stable, i.e., the same number of users are decoded. If multiple valid splits $\Pi_{\mathrm{valid}}$ exist, the algorithm selects one randomly to avoid deterministic local minima, and then the loop starts again to find a new split. If the network reaches a state where no further splits are possible, the algorithm will search for a merge operation. Initially, it will check if the number of RRHs in the active coalition has reached $K_{\mathrm{limit}}$. If this is true and all the users in the $\mathcal{U}_{\mathrm{t}}$ set are not yet decoded, it increases the $K_{\mathrm{limit}}$ by one. To add a new RRH to the coalition, the algorithm evaluates candidate partitions $\pi_{+r}$ formed by merging the active coalition with a single sleeping RRH $r \in\mathcal{R}\setminus \mathcal{R}_{\mathrm{on}}$. Since activating new RRHs is strictly penalized by the second element of the utility vector, the lexicographic operator mandates that a merge is only accepted if it strictly increases the number of decoded users. The algorithm computes the total social welfare of all potential merges, identifies the candidate that produces the maximum lexicographic improvement $W_{\max}$, and executes the merge. Again, if there are multiple partitions $\Pi_{\mathrm{best}}$ that produce the same maximum social welfare improvement, one is randomly selected. Furthermore, it should be noted that a split operation is inherently prioritized and tested each time a change happens in the network partition. This process continues until we reach a partition that meets the primary objective. Finally, there is a safeguard step that if there is no action that can improve the current state and $\mathcal{U}_{\mathrm{t}}$ is not met, the algorithm reverts back to the full initial network that forms $\pi_{\mathcal{R}}$. This process is summarized in Algorithm \ref{alg:slms}.

The initialization phase of the proposed algorithm establishes the initial network partition, denoted by $\pi_{\mathrm{seed}}$, by identifying a critical subset of mandatory RRHs, $\mathcal{R}_{\mathrm{crit}}$. To determine $\mathcal{R}_{\mathrm{crit}}$, we first compute the residual interference plus noise in each RRH $r \in \mathcal{R}$. This interference is defined as the noise power plus the total received power at $r$ excluding the signals from target users that include $r$ in their clusters and can potentially be canceled via SIC, expressed as
\begin{equation}
\omega_r = N_0 + P_{\mathrm{tot},r} - \sum_{\{v \in \mathcal{U}_{\mathrm{t}} \mid r \in \mathcal{C}_v\}} s_{v,r}.
\end{equation}
Using $\omega_r$, we evaluate the maximum achievable SINR for each target user across its candidate RRHs. If a user $u \in \mathcal{U}_{\mathrm{t}}$ satisfies the required SINR threshold at exactly one RRH that specific RRH represents the only decoding opportunity for the user. Consequently, this RRH is always required to be active in order for the network to be able to serve $\mathcal{U}_{\mathrm{t}}$ and is appended to $\mathcal{R}_{\mathrm{crit}}$. Finally, the initialization coalition structure $\pi_{\mathrm{seed}}$ includes the set $\mathcal{R}_{\mathrm{crit}}$ as the active coalition $\mathcal{R}_{\mathrm{on}}$ and each other RRH as singleton. In the event that the are no mandatory RRHs, i.e., $\mathcal{R}_{\mathrm{crit}} = \emptyset$, $\mathcal{R}_{\mathrm{on}}$ is initialized using a single randomly selected RRH. The complete initialization procedure is detailed in Algorithm \ref{alg:seed_init}.

\begin{figure}[!h]
    \centering
    \includegraphics[scale = 0.4]{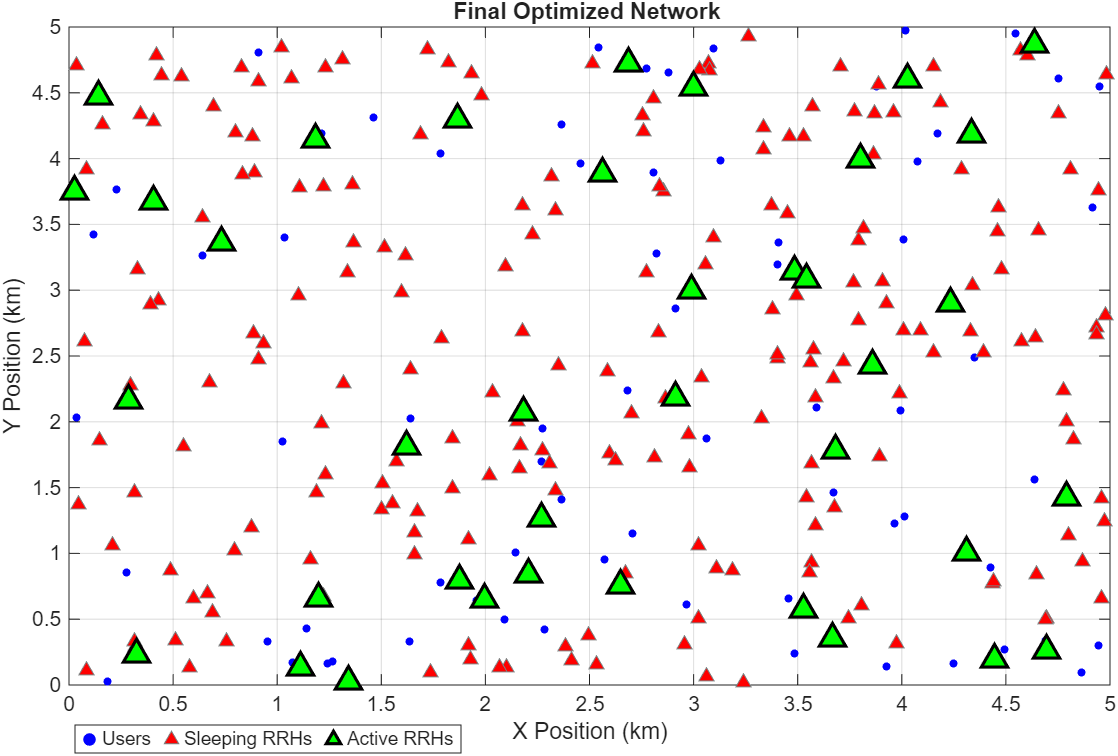}
    \caption{Final state of a network in a $5\times5$ km square with $\lambda_U=3$ users/$\text{km}^2$, $\lambda_{\mathrm{RRH}}=10$ RRHs/$\text{km}^2$, $\alpha=4$ and $\gamma_{\mathrm{th}}=6$ dB using Algorithm \ref{alg:slms}.}
    \label{game_example}
\end{figure}

In Fig. \ref{game_example}, we see the result using the stochastic lexicographic merge-and-split (SLMS) algorithm. In this specific example, all target users are decoded while deactivating $84.8 \%$ of the initial network.

Finally, the algorithm achieves a $\mathbb{D}_{hp}$ stability which is a weak equilibrium-like stability \cite{Coalition_Games}. It converges in a stable partition where no player has intention to perform a merge or a split operation. This type of stability can be considered as a kind of equilibrium with respect to merge-and-split. 

\section{Numerical Results and Simulations}
\subsection{Proposed Synergetic Decoding Algorithm}
In this subsection, the simulation results on the system performance are presented, when the algorithm proposed in Section \ref{ProposedAlgorithm} is implemented. In each case, we utilize a $10\times10$ km square to simulate a large-scale network. To mitigate edge effects, we apply the toroidal euclidean metric. The distance between the points $p_1=(x_1,y_1)$ and $p_2=(x_2,y_2)$ is given by
\begin{equation}
    d_{torus}(p_1,p_2)=\sqrt{d^2(x_1,x_2)+d^2(y_1,y_2)},
\end{equation}
where $d(x_1,x_2)=\min{(|x_1-x_2|,L-|x_1-x_2|)}$ and $L$ is the length of the side of the square. This method simulates an infinitely expanding homogeneous network using a finite area. It ensures that every RRH is subjected to an omnidirectional interference field and every user observes a symmetric, 360-degree distribution candidate serving RRHs. Also, we consider an SNR $\gamma_0=10$ dB and the RRH density is $\lambda_{\mathrm{RRH}}=10$  RRHs/$\text{km}^2$. Each point in the simulation curves is produced by Monte Carlo simulations with 10000 network realizations in the corresponding user density.

\begin{figure}[!h]
       \centering
       \begin{tikzpicture}
           \begin{semilogyaxis}[
           width = 0.85\linewidth,
           xlabel = {User Density (users/$\text{km}^2$)},
           ylabel = {Outage Probability},
           ymin = 0.00001,
           ymax = 1,
           xmin = 10,
           xmax = 25,
           grid = major,
           legend entries = {{DCMA},{No SIC + No Cooperation},{Local SIC}},
           legend cell align = {left},
           legend style = {font = \scriptsize},
           legend style={at={(1,0)},anchor=south east}
           ]

    \addplot[
    black,
    mark = x,
    mark repeat = 5,
    mark size = 3,
    mark phase = 0,
    line width = 1pt
    ]
    table {Paper_Data/Compare_3/duma_a_4_t_0_snr_10_smooth.dat};

    \addplot[
    black,
    mark = square,
    mark repeat = 2,
    mark size = 3,
    mark phase = 0,
    line width = 1pt
    ]
    table {Paper_Data/Compare_3/baseline_a_4_t_0_snr_10.dat};

    \addplot[
    black,
    mark = o,
    mark repeat = 2,
    mark size = 3,
    mark phase = 0,
    line width = 1pt
    ]
    table {Paper_Data/Compare_3/local_sic_a_4_t_0_snr_10.dat};
           \end{semilogyaxis}
       \end{tikzpicture}
       \vspace{-2mm}
       \caption{Outage probability vs user density for path loss exponent $\alpha=4$ and SINR threshold $\gamma_{\mathrm{th}}=0$ dB.}
       \vspace{-2mm}
       \label{a_4_compare_3}

   \end{figure}
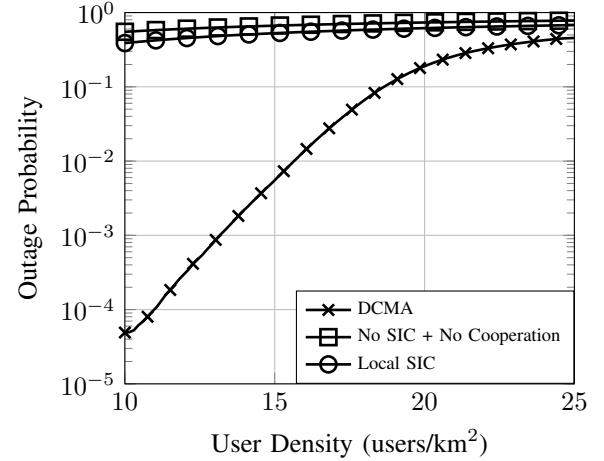 

\begin{figure}[!h]
       \centering
       \begin{tikzpicture}
           \begin{semilogyaxis}[
           width = 0.85\linewidth,
           xlabel = {User Density (users/$\text{km}^2$)},
           ylabel = {Outage Probability},
           ymin = 0,
           ymax = 1,
           xmin = 5,
           xmax = 15,
           xtick = {5,7,9,11,13,15},
           xticklabels = {5,7,9,11,13,15},
           grid = major,
           legend entries = {{DCMA},{No SIC + No Cooperation},{Local SIC}},
           legend cell align = {left},
           legend style = {font = \scriptsize},
           legend style={at={(1,0)},anchor=south east}
           ]

    \addplot[
    black,
    mark = x,
    mark repeat = 5,
    mark size = 3,
    mark phase = 0,
    line width = 1pt
    ]
    table {Paper_Data/Compare_3/duma_a_3_t_0_snr_10_smooth.dat};

    \addplot[
    black,
    mark = square,
    mark repeat = 2,
    mark size = 3,
    mark phase = 0,
    line width = 1pt
    ]
    table {Paper_Data/Compare_3/baseline_a_3_t_0_snr_10.dat};

    \addplot[
    black,
    mark = o,
    mark repeat = 2,
    mark size = 3,
    mark phase = 0,
    line width = 1pt
    ]
    table {Paper_Data/Compare_3/local_sic_a_3_t_0_snr_10.dat};
           \end{semilogyaxis}
       \end{tikzpicture}
       \vspace{-2mm}
       \caption{Outage probability vs user density for path loss exponent $\alpha=3$ and SINR threshold $\gamma_{\mathrm{th}}=0$ dB.}
       \vspace{-2mm}
       \label{a_3_compare_3}

   \end{figure}
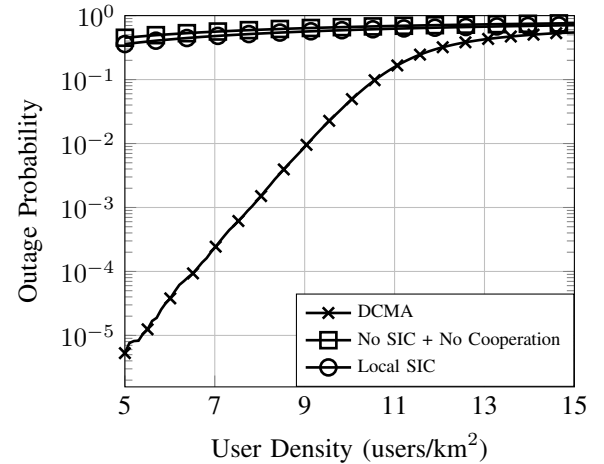 

The proposed system is compared with two baseline architectures. The first baseline assumes neither SIC nor cooperation among the RRHs. In this scheme, a user $i$ is successfully decoded if there exists at least one RRH $j$ where the SINR satisfies $\gamma_{i,j} \geq \gamma_{\mathrm{th}}$. The second baseline employs local SIC without cooperation. Here, a user is considered decoded if it is successfully resolved in the independent SIC cascade of at least one RRH. Both baselines evaluate the best outcome under their respective decoding scheme, ensuring a fair comparison. 

In Fig. \ref{a_4_compare_3} and Fig. \ref{a_3_compare_3}, the outage probability for a network with path loss exponents $\alpha=4$ and $\alpha=3$, respectively, is presented. In both cases, the SINR threshold is considered $\gamma_{\mathrm{th}}=0$ dB. It is obvious that the proposed DCMA significantly outperforms the benchmarks, especially in scenarios with lower user density, in which the proposed system is able to decode almost every user in the network, unlike the rest of the schemes which can decode less than half of the users. This superiority can be attributed to the fact that SIC is implemented in the proposed scheme, but most importantly to the cooperation between the RRHs, since decoding the message of a user in a single RRH suffices to be transmitted and considered known across the entire network, thus allowing every RRH to cancel its interference. Furthermore, it should be highlighted that the investigated system exhibits a counterintuitive behavior. Specifically, as the path loss exponent $\alpha$ increases, which leads to increased path loss attenuation, the system can support and decode higher user densities. This can be justified by considering that this causes higher channel disparities and channel diversity between the users and the RRHs which is known to facilitate the implementation of SIC because of higher interference suppression. The fact that there are multiple RRHs which can decode the message of each user and then share it with the rest of the network compensates for the increased path loss, enhancing the network performance. 


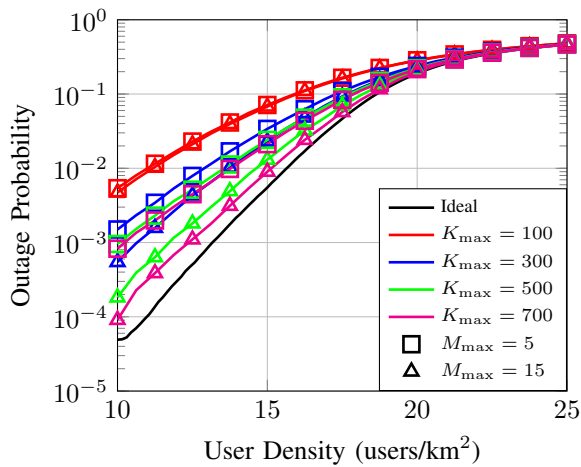
\begin{figure}[!h]
       \centering
       \begin{tikzpicture}
           \begin{semilogyaxis}[
           width = 0.85\linewidth,
           xlabel = {User Density (users/$\text{km}^2$)},
           ylabel = {Outage Probability},
           ymin = 1e-5,
           ymax = 1,
           xmin = 10,
           xmax = 25,
           grid = major,
           legend cell align = {left},
           legend style = {font = \scriptsize},
           legend style={at={(1,0)},anchor=south east}
           ]

    \addlegendimage{no markers, black, line width=1pt}
    \addlegendentry{Ideal}
    \addlegendimage{no markers, red, line width=1pt}
    \addlegendentry{$K_{\max}=100$}
    \addlegendimage{no markers, blue, line width=1pt}
    \addlegendentry{$K_{\max}=300$}
    \addlegendimage{no markers, green, line width=1pt}
    \addlegendentry{$K_{\max}=500$}
    \addlegendimage{no markers, magenta, line width=1pt}
    \addlegendentry{$K_{\max}=700$}
    \addlegendimage{only marks, mark=square, color=black, mark size=3, line width=1pt}
    \addlegendentry{$M_{\max}=5$}
    \addlegendimage{only marks, mark=triangle, color=black, mark size=3, line width=1pt}
    \addlegendentry{$M_{\max}=15$}

    \addplot[
    black,
    line width = 1pt,
    forget plot
    ]
    table {Paper_Data/Compare_3/duma_a_4_t_0_snr_10_smooth.dat};

    \addplot[
    red,
    mark = square,
    mark repeat = 2,
    mark size = 3,
    mark phase = 0,
    line width = 1pt,
    forget plot
    ]
    table {Paper_Data/Active_listener/active_a_4_t_0_K_100_M_5.dat};

    \addplot[
    red,
    mark = triangle,
    mark repeat = 2,
    mark size = 3,
    mark phase = 0,
    line width = 1pt,
    forget plot
    ]
    table {Paper_Data/Active_listener/active_a_4_t_0_K_100_M_15.dat};



    \addplot[
    blue,
    mark = square,
    mark repeat = 2,
    mark size = 3,
    mark phase = 0,
    line width = 1pt,
    forget plot
    ]
    table {Paper_Data/Active_listener/active_a_4_t_0_K_300_M_5.dat};

    \addplot[
    blue,
    mark = triangle,
    mark repeat = 2,
    mark size = 3,
    mark phase = 0,
    line width = 1pt,
    forget plot
    ]
    table {Paper_Data/Active_listener/active_a_4_t_0_K_300_M_15_pth_0_01.dat};


    \addplot[
    green,
    mark = square,
    mark repeat = 2,
    mark size = 3,
    mark phase = 0,
    line width = 1pt,
    forget plot
    ]
    table {Paper_Data/Active_listener/active_a_4_t_0_K_500_M_5_pth_0_001_bth_0_01.dat};

    \addplot[
    green,
    mark = triangle,
    mark repeat = 2,
    mark size = 3,
    mark phase = 0,
    line width = 1pt,
    forget plot
    ]
    table {Paper_Data/Active_listener/active_a_4_t_0_K_500_M_15_pth_0_001_bth_0_01.dat};


    \addplot[
    magenta,
    mark = square,
    mark repeat = 2,
    mark size = 3,
    mark phase = 0,
    line width = 1pt,
    forget plot
    ]
    table {Paper_Data/Active_listener/active_a_4_t_0_K_700_M_5_pth_0_001_bth_0_01.dat};

    \addplot[
    magenta,
    mark = triangle,
    mark repeat = 2,
    mark size = 3,
    mark phase = 0,
    line width = 1pt,
    forget plot
    ]
    table {Paper_Data/Active_listener/active_a_4_t_0_K_700_M_15_pth_0_001_bth_0_01.dat};
           \end{semilogyaxis}
       \end{tikzpicture}
       \vspace{-2mm}
       \caption{Proposed algorithm for path loss exponent $\alpha=4$ and SINR threshold $\gamma_{\mathrm{th}}=0$ dB.}
       \vspace{-2mm}
       \label{a_4_active}

   \end{figure}

\begin{figure}[!h]
       \centering
       \begin{tikzpicture}
           \begin{semilogyaxis}[
           width = 0.85\linewidth,
           xlabel = {User Density (users/$\text{km}^2$)},
           ylabel = {Outage Probability},
           ymin = 1e-6,
           ymax = 1,
           yminorticks = true,
           ytickten = {-6, -5, -4, -3, -2, -1, 0},
           xmin = 5,
           xmax = 15,
           xtick = {5,7,9,11,13,15},
           xticklabels = {5,7,9,11,13,15},
           grid = major,
           legend cell align = {left},
           legend style = {font = \scriptsize},
           legend style={at={(1,0)},anchor=south east}
           ]

    \addlegendimage{no markers, black, line width=1pt}
    \addlegendentry{Ideal}
    \addlegendimage{no markers, red, line width=1pt}
    \addlegendentry{$K_{\max}=100$}
    \addlegendimage{no markers, blue, line width=1pt}
    \addlegendentry{$K_{\max}=300$}
    \addlegendimage{no markers, green, line width=1pt}
    \addlegendentry{$K_{\max}=500$}
    \addlegendimage{no markers, magenta, line width=1pt}
    \addlegendentry{$K_{\max}=700$}
    \addlegendimage{only marks, mark=square, color=black, mark size=3, line width=1pt}
    \addlegendentry{$M_{\max}=5$}
    \addlegendimage{only marks, mark=triangle, color=black, mark size=3, line width=1pt}
    \addlegendentry{$M_{\max}=15$}

    \addplot[
    black,
    line width = 1pt,
    forget plot
    ]
    table {Paper_Data/Compare_3/duma_a_3_t_0_snr_10_smooth.dat};

    \addplot[
    red,
    mark = triangle,
    mark repeat = 2,
    mark size = 3,
    mark phase = 0,
    line width = 1pt,
    forget plot
    ]
    table {Paper_Data/Active_listener/active_a_3_t_0_K_100_M_15.dat};

    \addplot[
    red,
    mark = square,
    mark repeat = 2,
    mark size = 3,
    mark phase = 0,
    line width = 1pt,
    forget plot
    ]
    table {Paper_Data/Active_listener/active_a_3_t_0_K_100_M_5.dat};


    \addplot[
    blue,
    mark = square,
    mark repeat = 2,
    mark size = 3,
    mark phase = 0,
    line width = 1pt,
    forget plot
    ]
    table {Paper_Data/Active_listener/active_a_3_t_0_K_300_M_5_pth_0_01.dat};

    \addplot[
    blue,
    mark = triangle,
    mark repeat = 2,
    mark size = 3,
    mark phase = 0,
    line width = 1pt,
    forget plot
    ]
    table {Paper_Data/Active_listener/active_a_3_t_0_K_300_M_15_pth_0_01.dat};


    \addplot[
    green,
    mark = square,
    mark repeat = 2,
    mark size = 3,
    mark phase = 0,
    line width = 1pt,
    forget plot
    ]
    table {Paper_Data/Active_listener/active_a_3_t_0_K_500_M_5_pth_0_01.dat};
    
    \addplot[
    green,
    mark = triangle,
    mark repeat = 2,
    mark size = 3,
    mark phase = 0,
    line width = 1pt,
    forget plot
    ]
    table {Paper_Data/Active_listener/active_a_3_t_0_K_500_M_15_pth_0_01.dat};

    \addplot[
    magenta,
    mark = square,
    mark repeat = 2,
    mark size = 3,
    mark phase = 0,
    line width = 1pt,
    forget plot
    ]
    table {Paper_Data/Active_listener/active_a_3_t_0_K_700_M_5_pth_0_001_bth_0_01.dat};
    
    \addplot[
    magenta,
    mark = triangle,
    mark repeat = 2,
    mark size = 3,
    mark phase = 0,
    line width = 1pt,
    forget plot
    ]
    table {Paper_Data/Active_listener/active_a_3_t_0_K_700_M_15_pth_0_001_bth_0_01.dat};
    
           \end{semilogyaxis}
       \end{tikzpicture}
       \vspace{-2mm}
       \caption{Proposed algorithm for path loss exponent $\alpha=3$ and SINR threshold $\gamma_{\mathrm{th}}=0$ dB.}
       \vspace{-2mm}
       \label{a_3_active}

   \end{figure}
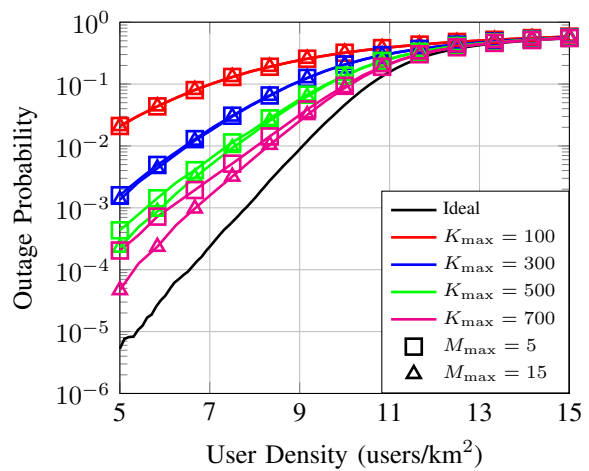   

Then we evaluate the performance of the proposed algorithm and the effect of the parameters. The key parameters that tune the performance are $K_{\max},M_{\max},\beta_{\mathrm{th}}$ and $P_{\mathrm{th}}$. The threshold $\beta_{\mathrm{th}}$ is the criterion that allows the RRHs to enter the cluster based on the distance to the user. The cluster size $K_{\max}$ and the sharing threshold $P_{\mathrm{th}}$ have a clear impact on interference mitigation because they dictate how far the interference of a specific user is canceled. A big cluster gives the opportunity for more RRHs to receive the decoded message and the threshold dictates how important the influence of a user is on an RRH. In addition, the parameter $M_{\max}$ plays an important role in the message sharing due to the fact that it corresponds to the number of RRHs a user can be decoded. Combined with the Active Listener sharing where a message is passed only to the RRHs that are needed for the remaining users, it can affect the feedback links.

\begin{table}[!h]
\vspace{-2mm}
\centering
\caption{Normalized threshold parameters $(\beta_{\mathrm{th}}/N_0,P_{\mathrm{th}}/N_0)$}
\resizebox{\columnwidth}{!}{%
\begin{tabular}{c||c|c||c|c}
\hline
 & \multicolumn{2}{c||}{\textbf{$\alpha=3$}} & \multicolumn{2}{c}{\textbf{$\alpha=4$}} \\ \cline{2-5}
\textbf{$K_{\max}$} & \textbf{$M_{\max}=5$} & \textbf{$M_{\max}=15$} & \textbf{$M_{\max}=5$} & \textbf{$M_{\max}=15$} \\ \hline \hline
100 & (0.1, 0.1) & (0.1, 0.1) & (0.1, 0.1) & (0.1, 0.1) \\
300 & (0.1, 0.01) & (0.1, 0.01) & (0.1, 0.1) & (0.1, 0.01) \\
500 & (0.1, 0.01) & (0.1, 0.01) & (0.01, 0.001) & (0.01, 0.001) \\
700 & (0.01, 0.001) & (0.01, 0.001) & (0.01, 0.001) & (0.01, 0.001) \\ \hline
\end{tabular}
}
\label{thresholds}
\end{table}

In Fig. \ref{a_4_active} and Fig. \ref{a_3_active}, we see the performance of the proposed algorithm with various parameters for networks with the same setup. Table \ref{thresholds} presents the selected threshold values normalized to the corresponding noise power $N_0$. It can be observed that the thresholds decrease as $K_{\max}$ increases, so that more RRHs fulfill the requirements to be part of the cluster of each user, and thus the additional capabilities of having larger clusters can be fully taken advantage of. As expected, increasing the cluster size $K_{\max}$ improves the performance of the system in terms of outage probability, as it allows the network to share the decoded messages to more RRHs leading to reduced interference. Furthermore, the effect of different $M_{\max}$ can be spotted in this figures. Specifically, a greater $M_{\max}$ reduces the outage probability of the system, since it allows more RRHs to be checked whether they can decode each user. However, it should be noted that the system benefits from a greater $M_{\max}$ when larger clusters are considered, since in these cases it is more probable that an RRH deeper in the candidate list can decode the user, due to the similar conditions. 

Comparing Fig. \ref{a_4_active} with Fig. \ref{a_3_active}, the effect of $\alpha$ in the performance of the system is verified. It is obvious that regardless of the cluster size, increased path loss exponent supports higher user densities with similar outage probabilities. In addition, in the case of $\alpha=3$ the system does not exhibit similar gains with the increase of $M_{\max}$, due to higher interference between users that necessitates RRHs with strong $s_{u,r}$. Thus, scenarios where RRHs that are sorted lower in $\mathcal{L}_u$ are chosen to decode user $u$ will occur more rarely. Finally, by adequately  tuning the parameters, i.e., increasing $K_{\max}$ and $M_{\max}$, which allows larger clusters and more feedback links in the system, the ideal performance can be better approached, as expected. It should be noted that this ideal performance is the optimal one and a lower bound in terms of outage for the proposed algorithm, since it is achieved by the exhaustive algorithm where each message is decoded in the RRH with the highest SINR and the decoded messages are shared across the entire network.



\begin{figure}[!h]
       \centering
       \begin{tikzpicture}
           \begin{axis}[
           width = 0.85\linewidth,
           xlabel = {User Density (users/$\text{km}^2$)},
           ylabel = {Feedback Links},
           ymin = 50,
           ymax = 800,
           ytick = {100,300,500,700},
           yticklabels = {100,300,500,700},
           xmin = 10,
           xmax = 25,
           grid = major,
           legend cell align = {left},
           legend columns = 2,
           legend style = {font = \scriptsize, fill=white, fill opacity= 0.5, text opacity=1},
           legend style={at={(0,1)},anchor=north west}
           ]

    \addlegendimage{no markers, red, line width=1pt}
    \addlegendentry{$K_{\max}=100$}
    \addlegendimage{no markers, blue, line width=1pt}
    \addlegendentry{$K_{\max}=300$}
    \addlegendimage{no markers, green, line width=1pt}
    \addlegendentry{$K_{\max}=500$}
    \addlegendimage{no markers, magenta, line width=1pt}
    \addlegendentry{$K_{\max}=700$}
    \addlegendimage{only marks, mark=square, color=black, mark size=3, line width=1pt}
    \addlegendentry{$M_{\max}=5$}
    \addlegendimage{only marks, mark=triangle, color=black, mark size=3, line width=1pt}
    \addlegendentry{$M_{\max}=15$}

    \addplot[
    red,
    mark = square,
    mark repeat = 2,
    mark size = 3,
    mark phase = 0,
    line width = 1pt,
    forget plot
    ]
    table {Paper_Data/Links/links_a_4_t_0_K_100_M_5.dat};

    \addplot[
    red,
    mark = triangle,
    mark repeat = 2,
    mark size = 3,
    mark phase = 0,
    line width = 1pt,
    forget plot
    ]
    table {Paper_Data/Links/links_a_4_t_0_K_100_M_15.dat};

    \addplot[
    blue,
    mark = square,
    mark repeat = 2,
    mark size = 3,
    mark phase = 0,
    line width = 1pt,
    forget plot
    ]
    table {Paper_Data/Links/links_a_4_t_0_K_300_M_5.dat};

    \addplot[
    blue,
    mark = triangle,
    mark repeat = 2,
    mark size = 3,
    mark phase = 0,
    line width = 1pt,
    forget plot
    ]
    table {Paper_Data/Links/links_a_4_t_0_K_300_M_15_pth_0_01.dat};

    \addplot[
    green,
    mark = square,
    mark repeat = 2,
    mark size = 3,
    mark phase = 0,
    line width = 1pt,
    forget plot
    ]
    table {Paper_Data/Links/links_a_4_t_0_K_500_M_5_pth_0_001_bth_0_01.dat};

    \addplot[
    green,
    mark = triangle,
    mark repeat = 2,
    mark size = 3,
    mark phase = 0,
    line width = 1pt,
    forget plot
    ]
    table {Paper_Data/Links/links_a_4_t_0_K_500_M_15_pth_0_001_bth_0_01.dat};

    \addplot[
    magenta,
    mark = square,
    mark repeat = 2,
    mark size = 3,
    mark phase = 0,
    line width = 1pt,
    forget plot
    ]
    table {Paper_Data/Links/links_a_4_t_0_K_700_M_5_pth_0_001_bth_0_01.dat};
    
    \addplot[
    magenta,
    mark = triangle,
    mark repeat = 2,
    mark size = 3,
    mark phase = 0,
    line width = 1pt,
    forget plot
    ]
    table {Paper_Data/Links/links_a_4_t_0_K_700_M_15_pth_0_001_bth_0_01.dat};

           \end{axis}
       \end{tikzpicture}
       \vspace{-2mm}
       \caption{Average feedback links per user for path loss exponent $\alpha=4$ and SINR threshold $\gamma_{\mathrm{th}}=0$ dB.}
       \vspace{-2mm}
       \label{a_4_links}

   \end{figure}
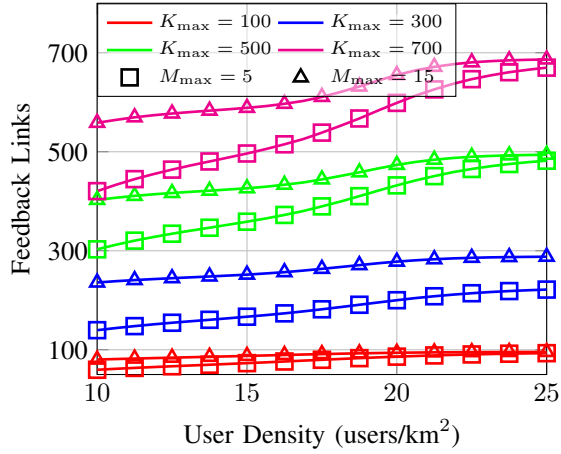

\begin{figure}[!h]
       \centering
       \begin{tikzpicture}
           \begin{axis}[
           width = 0.85\linewidth,
           xlabel = {User Density (users/$\text{km}^2$)},
           ylabel = {Feedback Links},
           ymin = 30,
           ymax = 800,
           ytick = {100,300,500,700},
           yticklabels = {100,300,500,700},
           xmin = 5,
           xmax = 15,
           xtick = {5,7,9,11,13,15},
           xticklabels = {5,7,9,11,13,15},
           grid = major,
           legend cell align = {left},
           legend columns = 2,
           legend style = {font = \scriptsize, fill=white, fill opacity= 0.5, text opacity=1},
           legend style={at={(0,1)},anchor=north west}
           ]

    \addlegendimage{no markers, red, line width=1pt}
    \addlegendentry{$K_{\max}=100$}
    \addlegendimage{no markers, blue, line width=1pt}
    \addlegendentry{$K_{\max}=300$}
    \addlegendimage{no markers, green, line width=1pt}
    \addlegendentry{$K_{\max}=500$}
    \addlegendimage{no markers, magenta, line width=1pt}
    \addlegendentry{$K_{\max}=700$}
    \addlegendimage{only marks, mark=square, color=black, mark size=3, line width=1pt}
    \addlegendentry{$M_{\max}=5$}
    \addlegendimage{only marks, mark=triangle, color=black, mark size=3, line width=1pt}
    \addlegendentry{$M_{\max}=15$}
   
    \addplot[
    red,
    mark = square,
    mark repeat = 2,
    mark size = 3,
    mark phase = 0,
    line width = 1pt,
    forget plot
    ]
    table {Paper_Data/Links/links_a_3_t_0_K_100_M_5.dat};

    \addplot[
    red,
    mark = triangle,
    mark repeat = 2,
    mark size = 3,
    mark phase = 0,
    line width = 1pt,
    forget plot
    ]
    table {Paper_Data/Links/links_a_3_t_0_K_100_M_15.dat};

    \addplot[
    blue,
    mark = square,
    mark repeat = 2,
    mark size = 3,
    mark phase = 0,
    line width = 1pt,
    forget plot
    ]
    table {Paper_Data/Links/links_a_3_t_0_K_300_M_5_pth_0_01.dat};
    
    \addplot[
    blue,
    mark = triangle,
    mark repeat = 2,
    mark size = 3,
    mark phase = 0,
    line width = 1pt,
    forget plot
    ]
    table {Paper_Data/Links/links_a_3_t_0_K_300_M_15_pth_0_01.dat};

    \addplot[
    green,
    mark = square,
    mark repeat = 2,
    mark size = 3,
    mark phase = 0,
    line width = 1pt,
    forget plot
    ]
    table {Paper_Data/Links/links_a_3_t_0_K_500_M_5_pth_0_01.dat};

    \addplot[
    green,
    mark = triangle,
    mark repeat = 2,
    mark size = 3,
    mark phase = 0,
    line width = 1pt,
    forget plot
    ]
    table {Paper_Data/Links/links_a_3_t_0_K_500_M_15_pth_0_01.dat};

    \addplot[
    magenta,
    mark = square,
    mark repeat = 2,
    mark size = 3,
    mark phase = 0,
    line width = 1pt,
    forget plot
    ]
    table {Paper_Data/Links/links_a_3_t_0_K_700_M_5_pth_0_001_bth_0_01.dat};
    
    \addplot[
    magenta,
    mark = triangle,
    mark repeat = 2,
    mark size = 3,
    mark phase = 0,
    line width = 1pt,
    forget plot
    ]
    table {Paper_Data/Links/links_a_3_t_0_K_700_M_15_pth_0_001_bth_0_01.dat};

           \end{axis}
       \end{tikzpicture}
       \vspace{-2mm}
       \caption{Average feedback links per user for path loss exponent $\alpha=3$ and SINR threshold $\gamma_{\mathrm{th}}=0$ dB.}
       \vspace{-2mm}
       \label{a_3_links}

   \end{figure}
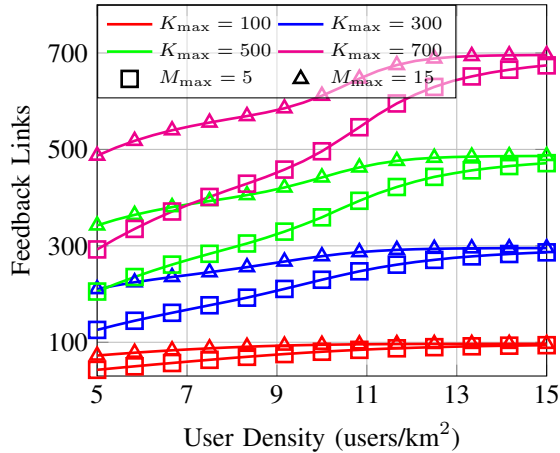

In Fig. \ref{a_4_links} and Fig. \ref{a_3_links}, the average number of feedback links utilized per user is depicted. As expected, higher values in these parameters increase the feedback links the network utilizes for the user decoding, as more RRHs are part of each user's cluster and may require the decoded message to be shared in order to perform SIC. In addition, the feedback links are increased with the user density of the network, especially for larger cluster sizes, because the increased interference due to the large number of users necessitates sharing the messages to more RRHs. In general, there is a trade-off between outage performance, CSI overhead and feedback link utilization. To improve the system performance, higher $K_{\max},M_{\max}$ and lower thresholds are required. This can achieve the desired outage probability at the expense of more CSI overhead due to larger clusters and more feedback link utilization. Thus, the algorithm parameters must be adjusted according to the communication requirements and the hardware and overhead capabilities the system has. 

\subsection{Network Optimization and SLMS Algorithm}
In this part, the results of the network optimization using the SLMS algorithm are presented. Specifically, we consider a $5\times5$ km square, an SNR $\gamma_0=10$ dB, and the RRH density $\lambda_{\mathrm{RRH}}=10$  RRHs/$\text{km}^2$. The SINR threshold is evaluated over the set $\gamma_{\mathrm{th}}\in \{3,6,9\}$ dB and the path loss exponent is $\alpha=4$. Regarding the parameters of the synergetic decoding algorithm, we consider $K_{\max}=100$, $M_{\max}=K_{\max}$, $\beta_{\mathrm{th}}=0.1 N_0$ and $P_{\mathrm{th}}=0$. Furthermore, every point in the following curves is produced after averaging 1000 network realizations.

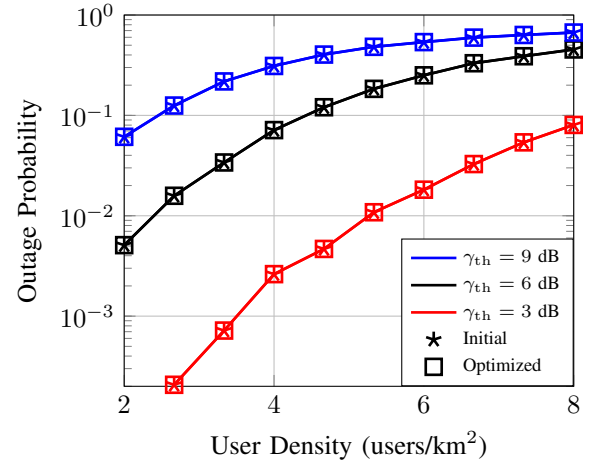
\begin{figure}[!h]
       \centering
       \begin{tikzpicture}
           \begin{semilogyaxis}[
           width = 0.85\linewidth,
           xlabel = {User Density (users/$\text{km}^2$)},
           ylabel = {Outage Probability},
           ymin = 0.0002,
           ymax = 1,
           xmin = 2,
           xmax = 8,
           grid = major,
           legend cell align = {left},
           legend style = {font = \scriptsize},
           legend style={at={(1,0)},anchor=south east}
           ]
           
    \addlegendimage{no markers, blue, line width=1pt}
    \addlegendentry{$\gamma_{\mathrm{th}}=9$ dB}
    \addlegendimage{no markers, black, line width=1pt}
    \addlegendentry{$\gamma_{\mathrm{th}}=6$ dB}
    \addlegendimage{no markers, red, line width=1pt}
    \addlegendentry{$\gamma_{\mathrm{th}}=3$ dB}
    \addlegendimage{only marks, mark=star, color=black, mark size=3, line width=1pt}
    \addlegendentry{Initial}
    \addlegendimage{only marks, mark=square, color=black, mark size=3, line width=1pt}
    \addlegendentry{Optimized}

    \addplot[
    black,
    mark = star,
    mark repeat = 1,
    mark size = 3,
    mark phase = 0,
    line width = 1pt
    ]
    table {Paper_Data/Game_theory/game_theory_a_4_t_6_initial_outage.dat};

    \addplot[
    black,
    mark = square,
    mark repeat = 1,
    mark size = 3,
    mark phase = 0,
    line width = 1pt
    ]
    table {Paper_Data/Game_theory/game_theory_a_4_t_6_slms_outage.dat};

    \addplot[
    blue,
    mark = star,
    mark repeat = 1,
    mark size = 3,
    mark phase = 0,
    line width = 1pt
    ]
    table {Paper_Data/Game_theory/game_theory_a_4_t_9_initial_outage.dat};

    \addplot[
    blue,
    mark = square,
    mark repeat = 1,
    mark size = 3,
    mark phase = 0,
    line width = 1pt
    ]
    table {Paper_Data/Game_theory/game_theory_a_4_t_9_slms_outage.dat};

    \addplot[
    red,
    mark = star,
    mark repeat = 1,
    mark size = 3,
    mark phase = 0,
    line width = 1pt
    ]
    table {Paper_Data/Game_theory/game_theory_a_4_t_3_initial_outage.dat};

    \addplot[
    red,
    mark = square,
    mark repeat = 1,
    mark size = 3,
    mark phase = 0,
    line width = 1pt
    ]
    table {Paper_Data/Game_theory/game_theory_a_4_t_3_slms_outage.dat};

           \end{semilogyaxis}
       \end{tikzpicture}
       \vspace{-2mm}
       \caption{Outage probability for the initial network and the optimized one using the SLMS algorithm.}
       \vspace{-2mm}
       \label{game_theory_outage}

   \end{figure} 

Fig. \ref{game_theory_outage} demonstrates the outage probability for both the initial network utilizing all the RRHs and the optimized network produced by the SLMS algorithm versus the user density of the network. As expected, increasing the user density, and thus the number of users, results in a higher outage probability, since the interference is increased. The simulation results reveal that the two performance curves are entirely indistinguishable. This precise alignment serves as a definitive proof that the SLMS algorithm identifies and isolates redundant nodes without introducing a degradation in the system's decoding performance by enforcing a strict lexicographic preference relation. Consequently, the optimized network achieves the exact same level of reliability with a significantly reduced resource utilization as shown in the following figure.

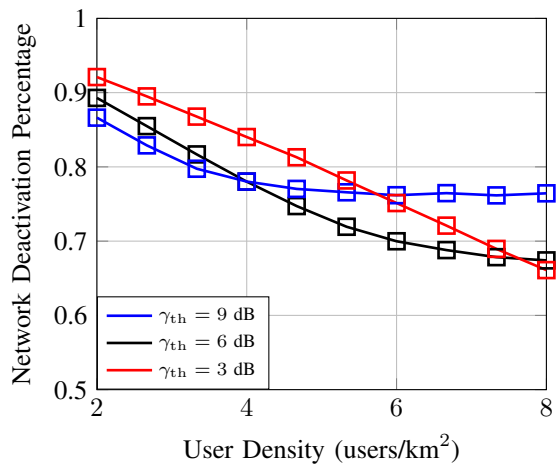
\begin{figure}[!h]
       \centering
       \begin{tikzpicture}
           \begin{axis}[
           width = 0.85\linewidth,
           xlabel = {User Density (users/$\text{km}^2$)},
           ylabel = {Network Deactivation Percentage},
           ymin = 0.5,
           ymax = 1,
           xmin = 2,
           xmax = 8,
           ytick = {0.5, 0.6, 0.7, 0.8, 0.9, 1},
           yticklabels = {0.5, 0.6, 0.7, 0.8, 0.9, 1},
           grid = major,
           legend cell align = {left},
           legend style = {font = \scriptsize},
           legend style={at={(0,0)},anchor=south west}
           ]

    \addlegendimage{no markers, blue, line width=1pt}
    \addlegendentry{$\gamma_{\mathrm{th}}=9$ dB}
    \addlegendimage{no markers, black, line width=1pt}
    \addlegendentry{$\gamma_{\mathrm{th}}=6$ dB}
    \addlegendimage{no markers, red, line width=1pt}
    \addlegendentry{$\gamma_{\mathrm{th}}=3$ dB}       

    \addplot[
    black,
    mark = square,
    mark repeat = 1,
    mark size = 3,
    mark phase = 0,
    line width = 1pt
    ]
    table {Paper_Data/Game_theory/game_theory_a_4_t_6_turned_off.dat};

    \addplot[
    blue,
    mark = square,
    mark repeat = 1,
    mark size = 3,
    mark phase = 0,
    line width = 1pt
    ]
    table {Paper_Data/Game_theory/game_theory_a_4_t_9_turned_off.dat};

    \addplot[
    red,
    mark = square,
    mark repeat = 1,
    mark size = 3,
    mark phase = 0,
    line width = 1pt
    ]
    table {Paper_Data/Game_theory/game_theory_a_4_t_3_turned_off.dat};

           \end{axis}
       \end{tikzpicture}
       \vspace{-2mm}
       \caption{Percentage of the network deactivated applying the SLMS algorithm.}
       \vspace{-2mm}
       \label{game_theory_turned_off}

   \end{figure}    

Fig. \ref{game_theory_turned_off} illustrates the percentage of deactivated RRHs achieved by the proposed SLMS algorithm. Remarkably, the algorithm successfully deactivates up to $92.1\%$ of the network nodes while strictly preserving the original system performance. It is observed that the network deactivation percentage initially exhibits a linear decrease with respect to user density, eventually flattening as the network becomes denser. This asymptotic behavior is fundamentally driven by the increasing system outage probability. Specifically, despite the increased number of active users in the network, the cardinality of the successfully decoded user subset, $|\mathcal{U}_{\mathrm{t}}|$, fails to scale proportionally with the total user count. Because further reductions in the deactivation percentage require the network to actually support and decode these additional users, the curve flattens when the network capacity is exhausted. It should be noted that the flat section of the curves corresponds to very high outage probability in which the system will not operate. Thus, the results demonstrate that a lower SINR threshold, $\gamma_{\mathrm{th}}$, yields superior performance in terms of network deactivation. By relaxing the decoding requirements, the algorithm can resolve users more efficiently, thereby maximizing the number of redundant RRHs that can be deactivated. This substantial reduction in active infrastructure translates to a profound improvement in energy efficiency and resource utilization for the proposed system.

Beyond minimizing active nodes, the proposed framework significantly reduces the utilization of feedback links. Fig. \ref{game_theory_links} compares the average number of feedback links required per decoded user between the initial grand coalition and the optimized network partition. The results demonstrate a substantial reduction and by minimizing these feedback exchanges, the optimized topology alleviates the fronthaul payload and greatly simplifies both the synchronization and routing of decoded messages. Furthermore, it is worth noting that even in the initial, full network, the number of feedback links remains strictly below the upper bound $K_{\max}$, despite the system being configured to share the decoded messages across the entire user's cluster. This baseline efficiency is a direct consequence of the Active Listener technique, once again underscoring its critical role in mitigating network overhead.

\begin{figure}[!h]
       \centering
       \begin{tikzpicture}
           \begin{axis}[
           width = 0.85\linewidth,
           xlabel = {User Density (users/$\text{km}^2$)},
           ylabel = {Feedback Links},
           ymin = 0,
           ymax = 100,
           xmin = 2,
           xmax = 8,
           grid = major,
           legend cell align = {left},
           legend style = {font = \scriptsize},
           legend style={at={(1,0.5)},anchor=south east}
           ]

    \addlegendimage{no markers, blue, line width=1pt}
    \addlegendentry{$\gamma_{\mathrm{th}}=9$ dB}
    \addlegendimage{no markers, black, line width=1pt}
    \addlegendentry{$\gamma_{\mathrm{th}}=6$ dB}
    \addlegendimage{no markers, red, line width=1pt}
    \addlegendentry{$\gamma_{\mathrm{th}}=3$ dB}
    \addlegendimage{only marks, mark=star, color=black, mark size=3, line width=1pt}
    \addlegendentry{Initial}
    \addlegendimage{only marks, mark=square, color=black, mark size=3, line width=1pt}
    \addlegendentry{Optimized}

    \addplot[
    black,
    mark = star,
    mark repeat = 1,
    mark size = 3,
    mark phase = 0,
    line width = 1pt
    ]
    table {Paper_Data/Game_theory/game_theory_a_4_t_6_initial_links.dat};

    \addplot[
    black,
    mark = square,
    mark repeat = 1,
    mark size = 3,
    mark phase = 0,
    line width = 1pt
    ]
    table {Paper_Data/Game_theory/game_theory_a_4_t_6_slms_links.dat};

    \addplot[
    blue,
    mark = star,
    mark repeat = 1,
    mark size = 3,
    mark phase = 0,
    line width = 1pt
    ]
    table {Paper_Data/Game_theory/game_theory_a_4_t_9_initial_links.dat};

    \addplot[
    blue,
    mark = square,
    mark repeat = 1,
    mark size = 3,
    mark phase = 0,
    line width = 1pt
    ]
    table {Paper_Data/Game_theory/game_theory_a_4_t_9_slms_links.dat};

    \addplot[
    red,
    mark = star,
    mark repeat = 1,
    mark size = 3,
    mark phase = 0,
    line width = 1pt
    ]
    table {Paper_Data/Game_theory/game_theory_a_4_t_3_initial_links.dat};

    \addplot[
    red,
    mark = square,
    mark repeat = 1,
    mark size = 3,
    mark phase = 0,
    line width = 1pt
    ]
    table {Paper_Data/Game_theory/game_theory_a_4_t_3_slms_links.dat};

           \end{axis}
       \end{tikzpicture}
       \vspace{-2mm}
       \caption{Comparison of average feedback links per user utilized between the initial network and the optimized one using the SLMS algorithm.}
       \vspace{-2mm}
       \label{game_theory_links}

   \end{figure}
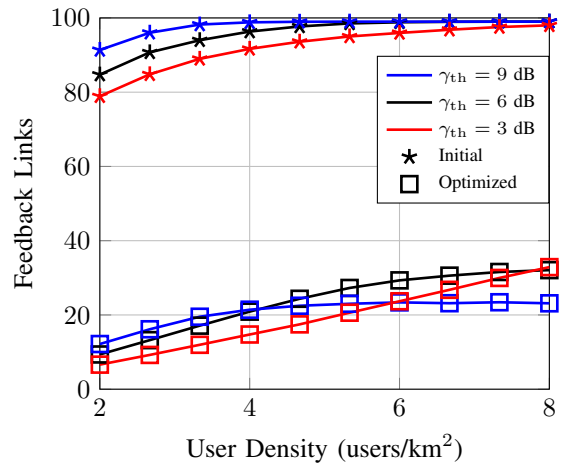 

Another important result that we should mention is the fact that the safeguard step that reverts the network to the $\pi_{\mathcal{R}}$ partition was never once activated in all the simulation procedure. This highlights the convergence and stability capabilities that the SLMS algorithm has. Finally, a not obvious network property has been proved to exist. The results confirm that there is extensive redundancy, allowing optimal performance to be matched using significantly fewer active RRHs.  However, it is critical to emphasize that physically deploying a network with a lower baseline RRH density, $\lambda_{\mathrm{RRH}}$, would fail to achieve this same performance.  The amount of initial RRHs needed remains the same and the reduced partition will change dynamically according to the user locations and their channel state. 

\section{Conclusions}
In this paper, we introduced the DCMA framework with large-scale, cell-free C-RAN architecture. By leveraging stochastic geometry, we established a system model for networks where RRHs decode and cooperatively share user messages via feedback links. A primary contribution of this work is the development of a novel synergetic decoding algorithm. Our evaluations demonstrate that this proposed method effectively resolves complex user assignment and message routing challenges while adhering to practical network constraints. Specifically, by eliminating the need for global CSI and network-wide message broadcasting, the algorithm ensures practical feasibility. Furthermore, our comparative analysis definitively illustrates the significant superiority of this cooperative, SIC-enabled approach over non-cooperative or non-SIC baselines. Importantly, we highlighted the crucial trade-off between overall system performance, CSI acquisition overhead, and feedback link utilization. To further enhance network efficiency, a game theoretic approach was applied to minimize RRH utilization. Through the development of a dynamic coalition formation game and a merge-and-split algorithm, the proposed framework substantially reduced both the use of RRHs and feedback links without compromising the underlying decoding performance. Ultimately, these advancements provide a robust blueprint for the scalable design, practical deployment, and optimization of future 6G DCMA networks.  


\bibliographystyle{ieeetr}
\bibliography{Paper_Bibliography}
\end{document}